\documentclass[12pt]{article}

\usepackage{acro}
\usepackage{amsmath}
\usepackage{amssymb}
\usepackage{amsthm}
\usepackage{mathtools}
\usepackage{bm}
\usepackage{color}
\usepackage{fullpage}
\usepackage{graphicx}
\usepackage[hidelinks]{hyperref}
\usepackage{natbib}
\usepackage{scalerel}
\usepackage{tabularx}
\usepackage{url}
\usepackage{cleveref}
\usepackage{verbatim}

\usepackage{minitoc}

\newtheorem{theorem}{Theorem}

\newtheorem{lemma}[theorem]{Lemma}
\newtheorem{proposition}[theorem]{Proposition}
\newtheorem{corollary}[theorem]{Corollary}
\newtheorem{example}[theorem]{Example}
\newtheorem{remark}[theorem]{Remark}

\DeclareAcronym{GP}{short=GP, long=Gaussian process, long-plural = \textit{es}}
\DeclareAcronym{MFM}{short=MFM, long=multi-fidelity modelling}
\DeclareAcronym{LOFI}{short=lo-fi, long=low fidelity}
\DeclareAcronym{HIFI}{short=hi-fi, long=high fidelity}
\DeclareAcronym{MCMC}{short=MCMC, long=Markov chain Monte Carlo}
\DeclareAcronym{RKHS}{short=RKHS, long=reproducing kernel Hilbert space}
\DeclareAcronym{GRE}{short=GRE, long=Gauss--Richardson Extrapolation}
\acsetup{format/first-long=\itshape}

\DeclareMathOperator*{\argmax}{arg\,max}

\begin{document}

\doparttoc
\faketableofcontents 

\title{Probabilistic Richardson Extrapolation}
\author{Chris. J. Oates$^1$, Toni Karvonen$^2$, Aretha L. Teckentrup$^3$ \\
Marina Strocchi$^{4,5}$, Steven A. Niederer$^{4,5,6}$ \\
\footnotesize $^1$ School of Mathematics, Statistics and Physics, Newcastle University, UK \\
\footnotesize $^2$ Department of Mathematics and Statistics, University of Helsinki, FI \\
\footnotesize $^3$ School of Mathematics and Maxwell Institute for Mathematical Sciences, University of Edinburgh, UK \\
\footnotesize $^4$ National Heart and Lung Institute, Imperial College London, UK \\
\footnotesize $^5$ School of Biomedical Engineering and Imaging Sciences, King’s College London, UK \\
\footnotesize $^6$ The Alan Turing Institute, UK
}
\maketitle

\begin{abstract}
For over a century, extrapolation methods have provided a powerful tool to improve the convergence order of a numerical method.
However, these tools are not well-suited to modern computer codes, where multiple continua are discretised and convergence orders are not easily analysed.
To address this challenge we present a probabilistic perspective on Richardson extrapolation, a point of view that unifies classical extrapolation methods with modern multi-fidelity modelling, and handles uncertain convergence orders by allowing these to be statistically estimated.
The approach is developed using Gaussian processes, leading to \emph{Gauss--Richardson Extrapolation} (GRE).
Conditions are established under which extrapolation using the conditional mean achieves a polynomial (or even an exponential) speed-up compared to the original numerical method.
Further, the probabilistic formulation unlocks the possibility of experimental design, casting the selection of fidelities as a continuous optimisation problem which can then be (approximately) solved.
A case-study involving a computational cardiac model demonstrates that practical gains in accuracy can be achieved using the GRE method.
\end{abstract}

\section{Introduction}
\label{sec: introduction}

Testing of hypotheses underpins the scientific method, and increasingly these hypotheses are model-based. 
Deterministic or stochastic mathematical models are routinely used to represent mechanisms hypothesised to govern diverse phenomena, such as aerodynamics or electrochemical regulation of the human heart.
In these cases, critical scientific enquiry demands a comparison of the model against a real-world dataset. 
The practical challenge is two-fold; to simulate from the mathematical model, and to obtain a real-world dataset.
Here we focus on the first challenge -- simulating from the model -- which can be arbitrarily difficult depending on the complexity of the model.
For example, simulating a single cycle of a jet engine to an acceptable numerical precision routinely requires $10^6$ core hours \citep{arroyo2021towards}, while accurate simulation from the cardiac models that we consider later in this paper at steady state requires $10^4$ core hours in total \citep{strocchi2023cell}.
To drive progress in these, and many other diverse scientific domains, there is an urgent need for statistical and computational methodology that can mitigate the high cost of accurately simulating from a mathematical model.

Abstractly, we enumerate all of the \emph{discretisation parameters} involved in approximate simulation from the mathematical model using scalars $\mathbf{x} = (x_1,\dots,x_d)$, such that each component of $\mathbf{x}$ controls the error due to a particular aspect of discretisation; for example, $x_1$ could be a time step size, $x_2$ could be the width of a spatial mesh, and $x_3$ could be an error tolerance for an adaptive numerical method.
The principal requirement is that the ideal mathematical model corresponds to the limit $\mathbf{x} \rightarrow \mathbf{0}$ where no discretisation is performed.
Given a value of $\mathbf{x}$, we denote as $f(\mathbf{x})$ the associated numerical approximation to the continuum quantity $f(\mathbf{0})$ from the mathematical model.
The computational cost of such an evaluation will be denoted $c(\mathbf{x})$, with $c(\mathbf{0}) = \infty$ being typical.
The computational challenge addressed in this paper is to produce an accurate approximation to $f(\mathbf{0})$, based on a dataset of simulations $\{ f(\mathbf{x}_j) \}$, where $\{\mathbf{x}_j\} \subset (0,\infty)^d$, such that the computational cost of obtaining $\{f(\mathbf{x}_j)\}$ remains within a prescribed budget.
For this initial discussion we focus on scalar-valued model output, but we generalise to multivariate and infinite-dimensional model output in \Cref{subsec: multivar output}.

Several solutions have been proposed to perform approximate simulation at reduced cost.
In what follows it is useful to draw a distinction between \emph{extrapolation methods}, applicable to the situation where a mathematical model is discretised for simulation and numerical analysis of the discretisation error can be performed, and \emph{modern solutions} that are typically applied to `black box' computer codes for which numerical analysis is impractical.

\paragraph{Extrapolation Methods}

A unified presentation of extrapolation methods, that includes the most widely-used algorithms, is provided by the so-called \emph{E-algorithm} \citep[see the survey of][]{brezinski1989survey}.
The starting point is a (real-valued) convergent sequence, which in our setting we interpret as a sequence of numerical approximations $(f(\mathbf{x}_m))_{m \in \mathbb{N}}$, where $\mathbf{x}_m$ is a vector of discretisation parameters controlling the error in approximating the mathematical model, while $f(\mathbf{0})$ represents the continuum quantity of interest.
The E-algorithm posits an \emph{ansatz} that
\begin{align}
    f(\mathbf{x}_m) = f(\mathbf{0}) + a_1 g_1(m) + \dots + a_{n-1} g_{n-1}(m) \label{eq: ansatz}
\end{align}
for some unknown $a_1,\dots,a_{n-1} \in \mathbb{R}$, some known functions $g_i : \mathbb{N} \rightarrow \mathbb{R}$, and all $m \in \mathbb{N}$.
Then, instantiating \eqref{eq: ansatz} for $m, m+1 , \dots , m + n - 1$, we may solve for the unknown $a_1,\dots,a_{n-1}$ and $f(\mathbf{0})$ in terms of the $n$ values $f(\mathbf{x}_m), \dots, f(\mathbf{x}_{m+n-1})$.
Indeed, solving this linear system for $f(\mathbf{0})$ leads to the estimator
\begin{align}
    S_m \coloneqq S(f(\mathbf{x}_m),\dots,f(\mathbf{x}_{m+n-1})) = \frac{ \left| \begin{array}{ccc} f(\mathbf{x}_m) & \dots & f(\mathbf{x}_{m+n-1}) \\ g_1(m) & \dots & g_1(m+n-1) \\ \vdots & & \vdots \\ g_n(m) & \dots & g_n(m+n-1) \end{array} \right| }{ \left| \begin{array}{ccc} 1 & \dots & 1 \\ g_1(m) & \dots & g_1(m+n-1) \\ \vdots & & \vdots \\ g_n(m) & \dots & g_n(m+n-1) \end{array} \right| } . \label{eq: S}
\end{align}
Under appropriate assumptions, the sequence $(S_m)_{m \in \mathbb{N}}$ constructed based on $(f(\mathbf{x}_m))_{m \in \mathbb{N}}$ as in \eqref{eq: S} not only has the same limit, $f(\mathbf{0})$, but also converges to that limit faster in the sense that $\lim_{m \rightarrow \infty} (S_m - f(\mathbf{0})) / (f(\mathbf{x}_m) - f(\mathbf{0})) = 0$; for precise statements see Chapter 2 of \citet{brezinski2013extrapolation}.

The principal classes of extrapolation methods concern either the case of a single discretisation parameter $x_m$, or they maintain ambivalence about $\mathbf{x}_m$ by operating only on the values of the sequence $(f(\mathbf{x}_m))_{m \in \mathbb{N}}$.
In either case, different extrapolation methods correspond to different basis functions $g_i$ in \eqref{eq: ansatz}. 
\emph{Richardson extrapolation} corresponds to $g_i(m) = x_m^i$, in which case \eqref{eq: ansatz} is recognised as polynomial extrapolation to the origin \citep{richardson1911ix,lf1927deferred}.
The existence of a Taylor expansion of $f$ at the origin is sufficient to guarantee a polynomial-rate convergence acceleration using Richardson's method.
Other examples of extrapolation methods include \emph{Shanks' transformation} $g_i(m) = f(\mathbf{x}_{m+i}) - f(\mathbf{x}_{m+i-1})$ \citep{shanks1955non}, the \emph{Germain--Bonne transformation} $g_i(m) = (f(\mathbf{x}_{m+1}) - f(\mathbf{x}_m))^i$ \citep{germain1990convergence}, and \emph{Thiele's rational extrapolation method} $g_i(m) = x_m^i$, $g_{i+p}(m) = f(x_m) x_m^i$ for $i = 1,\dots,p$, $n = 2p + 1$ \citep{thiele1909interpolationsrechnung,bulirsch1964fehlerabschatzungen,larkin1967some}.
A careful numerical analysis of $f$ is usually required to determine when a particular extrapolation method can be applied.
To the best of our knowledge, ideas from statistics and uncertainty quantification do not feature prominently, if at all, in the literature on extrapolation methods.
In addition, the question of how best to construct the sequence $(\mathbf{x}_m)_{m \in \mathbb{N}}$ under a constraint on the overall computational budget does not appear to have been systematically addressed.
Further background can be found in the book-level treatment of \citet{sidi2003practical} and~\citet{brezinski2013extrapolation}.

Though rather classical, extrapolation methods continue to find new and useful applications, including in optimal transport \citep{chizat2020faster}, regularisation and training of machine learning models \citep{bach2021effectiveness}, and sampling with Markov chain Monte Carlo \citep{durmus2016stochastic}.

\paragraph{Modern Solutions}

If the mathematical model additionally involves one or more degrees of freedom $\bm{\theta}$, numerical approximations $f_{\bm{\theta}}(\mathbf{x})$ are often required across a range of values for $\bm{\theta}$ to identify configurations that are consistent with observations from the real world.
Since the introduction of additional degrees of freedom further complicates numerical analysis, this setting has motivated the development of black box methods that can be applied in situations where numerical analysis is impractical.
Among these, \emph{emulation} and \emph{multi-fidelity modelling} are arguably the most prominent.

In \emph{emulation} one attempts to approximate the map $\bm{\theta} \mapsto f_{\bm{\theta}}(\mathbf{x}_{\text{hi-fi}})$, where the discretisation parameters $\mathbf{x}_{\text{hi-fi}}$ are typically fixed and correspond to a suitably \ac{HIFI} model.
This enables prediction of computer code output at values of $\bm{\theta}$ for which simulation was not performed \citep{sacks1989design}.
A variety of sophisticated techniques have been developed to identify an appropriate basis or subspace in which an emulator can be constructed, such as \textit{reduced order modelling} \citep{lucia2004reduced}.
One drawback of emulation is that it can be \textit{data hungry}; in applications for which it is only possible to perform a small number $n$ of simulations, and for which insight from numerical analysis is unavailable, one usually cannot expect to obtain high-quality predictions.
A second drawback is that emulation treats the discretised model $\bm{\theta} \mapsto f_{\bm{\theta}}(\mathbf{x}_{\text{hi-fi}})$ as the target, whereas in reality the continuum mathematical model $\bm{\theta} \mapsto f_{\bm{\theta}}(\mathbf{0})$ is of principal interest.

A partial solution to the drawbacks of emulation is \ac{MFM}, in which one supplements a small number of simulations from the \ac{HIFI} model $\bm{\theta} \mapsto f_{\bm{\theta}}(\mathbf{x}_{\text{hi-fi}})$ with a larger number of simulations from one or more cheaper \ac{LOFI} models $\bm{\theta} \mapsto f_{\bm{\theta}}(\mathbf{x}_{\text{lo-fi}})$ \citep{peherstorfer2018survey}.
\Ac{LOFI} models can sometimes be obtained using coarse-grid approximations, early stopping of iterative algorithms, or linearisation \citep{piperni2013development}.
Alternatively, \ac{LOFI} models could involve only a subset of the relevant physical mechanisms, an approach popular e.g. in climate science \citep{held2005gap,majda2010quantifying}.
Once specified, the models of different fidelities can be combined in different ways: one can either use the \ac{HIFI} model to periodically `check' (and possibly adapt) the \ac{LOFI} models; or one can use the \ac{LOFI} models as pilot runs to decide whether or not to evaluate the \ac{HIFI} model; or one can use the information from all models simultaneously, by defining a multi-fidelity surrogate model \citep{craig1998constructing,kennedy2000predicting,cumming2009small,ehara2023adaptive} where correlation between models is taken into account. 
Provided that the \ac{LOFI} models are correlated with the original model, these additional cheap simulations can be leveraged to more accurately predict computer code output.
The principal drawback of multi-fidelity modelling is that there is limited guidance on how the \ac{LOFI} models should be constructed, and a poor choice can fail to improve (or even worsen) predictive performance, whilst incurring an additional computational cost.
In addition, as with emulation, the literature on \ac{MFM} tends to treat the \ac{HIFI} model as the target, rather than the continuum mathematical model.

\paragraph{Other Related Work}

Some alternative lines of research will briefly be discussed.
\emph{Probabilistic numerics} casts numerical approximation as a statistical task \citep{hennig2015probabilistic}, with Bayesian principles used to quantify uncertainty regarding the continuum model of interest \citep{cockayne2019bayesian}.
However, the focus of that literature is the design of numerical methods, in contrast to extrapolation methods which operate on the output of existing numerical methods.
In parallel, the application of machine learning methods to numerical tasks has received recent attention; for example deep learning is being used for numerical approximation of high-dimensional parametric partial differential equations \citep{han2018solving}.
This literature does not attempt extrapolation as such, with a \ac{HIFI} numerical method typically used to provide a training dataset.
Gaussian processes have been used in specific applications to extrapolate a series of numerical approximations to a continuum quantity of interest $f(\mathbf{0})$, for example in \citet{thodoroff2023multi} to model ice sheets in Antarctic, and in \citet{ji2022conglomerate} to model the evolution of the quark-gluon plasma following the Big Bang.
To date, however, convergence acceleration has not been studied in the Gaussian process context.
An important numerical task encountered in statistics is to approximate an expected value of interest $f(\mathbf{0}) = \mathbb{E}[X(\mathbf{0})]$.
Unbiased estimation of $f(\mathbf{0})$ at finite cost is possible in this setting using the methodology of \citet{rhee2015unbiased}, provided one can construct a sequence $(X(\mathbf{x}_n))_{n \in \mathbb{N}}$ of computable stochastic approximations to $X(\mathbf{0})$, such that the variance of $X(\mathbf{x}_n) - X(\mathbf{x}_{n-1})$ decays sufficiently fast. 
Similar de-biasing ideas have since been used in the context of Markov chain Monte Carlo \citep{jacob2020unbiased}.
Multilevel methods, based on such sequences, have been combined with Richardson extrapolation in \citet{lemaire2017multilevel,beschle2023quasi}.

\paragraph{Our Contribution}

This paper proposes a probabilistic perspective on extrapolation methods that unifies extrapolation methods and \ac{MFM}.
The approach is instantiated using a \emph{numerical analysis-informed Gaussian process} to approximate the map $\mathbf{x} \mapsto f(\mathbf{x})$, as described in \Cref{sec: methods}, where the conditional mean can be interpreted as a (novel) extrapolation method, in the sense that it provably achieves a polynomial (or even an exponential) speed-up compared to the original numerical method.
Like Richardson extrapolation, our theoretical arguments are rooted in Taylor expansions, so the name \ac{GRE} is adopted.
The probabilistic formulation of extrapolation methods confers several advantages:
\begin{itemize}
    \item In contrast to classical extrapolation methods, which focus on the case of a univariate discretisation parameter $x_n$, it is straight-forward to consider a vector of discretisation parameters $\mathbf{x}_n$ within a regression framework.
    In \Cref{subsec: prior,subsec: fitting the GP,subsec: justify prior} the probabilistic approach is laid out, then in \Cref{subsec: faster convergence,subsec: continuous and discrete} higher-order convergence guarantees for \ac{GRE} are established.
    \item Credible sets for the continuum quantity of interest $f(\mathbf{0})$ can be constructed, enabling computational uncertainty to be integrated into experimental design and downstream decision-support.
    The asymptotic performance of \ac{GRE} credible sets is analysed in \Cref{subsec: UQ}.
    \item In contrast to existing approaches in \ac{MFM}, where a discrete set of fidelities are specified at the outset, \ac{GRE} operates on a continuous spectrum of fidelities and casts the selection of fidelities as a cost-constrained experimental design problem, which can then be approximately solved using methods described in \Cref{subsec: design}.
    \item For computer models whose convergence order is difficult to analyse, the probabilistic formulation allows for convergence orders to be formally estimated.
    The consistency of a maximum quasi-likelihood approach to estimating unknown convergence order is established in \Cref{subsec: estimate order}.
\end{itemize}
The methodology is rigorously tested in the context of simulating from a computational cardiac model involving separate spatial and temporal discretisation parameters in \Cref{sec: applications}.
The sensitivity of the cardiac model to the different discretisation parameters is first estimated from \ac{LOFI} simulations, then an optimal experimental design is generated and used to estimate the true trajectory of the cardiac model in the continuum limit.
Our experimental results demonstrate that a practical gain in accuracy can be achieved with our \ac{GRE} method.
Though our assessment focuses on a specific cardiac model of scientific and clinical interest, the methodology is general and offers an exciting possibility to accelerate computation in the diverse range of scenarios in which computationally-intensive simulation is performed. 
A closing discussion is contained in \Cref{sec: discuss}.

%

\section{Methodology}
\label{sec: methods}

This section presents a novel probabilistic perspective on extrapolation methods, which we instantiate using \acp{GP} to produce methodology that we term \acf{GRE}.
For simplicity of presentation, we first consider the case of a scalar quantity of interest, generalising to arbitrary-dimensional quantities of interest in \Cref{subsec: multivar output}.

\smallskip 

\paragraph{Set-Up} 
Let $f : \mathcal{X} \rightarrow \mathbb{R}$ be a (non-random) real-valued function on a bounded set $\mathcal{X} \subset [0,\infty)^d$ such that $\mathbf{0} \in \mathcal{X}$.
As we explained in \Cref{sec: introduction}, the output $f(\mathbf{x})$ for $\mathbf{x} \neq \mathbf{0}$ will represent a numerical approximation to a continuum quantity $f(\mathbf{0})$ of interest, and for extrapolation to be possible at all we will minimally need to assume that $f$ is continuous at $\mathbf{0}$.

\paragraph{Notation}  
For $\bm{\beta} \in \mathbb{N}_0^p$, let $\partial^{\bm{\beta}} g$ denote the mixed partial derivative $\mathbf{x} \mapsto \partial_{x_1}^{\beta_1} \dots \partial_{x_p}^{\beta_p} \, g(\mathbf{x})$ of a function $g: D \rightarrow \mathbb{R}$, whenever this is well-defined and $D \subseteq \mathbb{R}^p$.
Let $C^s(D)$ denote the set of $s$-times continuously differentiable functions $g : D \rightarrow \mathbb{R}$, meaning that $\partial^{\bm{\beta}} g$ is continuous for all $|\bm{\beta}| \leq s$, where $|\bm{\beta}| \coloneqq \beta_1 + \cdots + \beta_p$.
For $g: D \rightarrow \mathbb{R}$ bounded, let $\|g\|_{L^\infty(D)} \coloneqq \sup_{\mathbf{x} \in D} |g(\mathbf{x})|$.
Let $\pi_r(D)$ denote the set of all polynomial functions of total degree at most $r$ on $D$.
For vectors $\mathbf{a},\mathbf{b} \in \mathbb{R}^d$, let $[\mathbf{a},\mathbf{b}] \coloneqq [a_1,b_1] \times \dots \times [a_d,b_d] \subset \mathbb{R}^d$, and similarly for $[\mathbf{a},\mathbf{b})$ and so forth.
Let $\mathcal{GP}(m,k)$ denote the law of a \Ac{GP} with mean function $m$ and covariance function $k$; background on \Acp{GP} can be found in \citet{rasmussen2006gaussian}.

\subsection{A Numerical Analysis-Informed Gaussian Process}
\label{subsec: prior}

Assuming for the moment that numerical analysis of $\mathbf{x} \mapsto f(\mathbf{x})$ can be performed, our first aim is to encode the resulting bounds on discretisation error into a statistical regression model.
Training such a \emph{numerical analysis-informed} regression model on data $\{f(\mathbf{x}_i)\}_{i=1}^n$ obtained at distinct inputs $X_n = \{\mathbf{x}_i\}_{i=1}^n \subset \mathcal{X} \setminus \{\mathbf{0}\}$ enables statistical prediction of the limit $f(\mathbf{0})$, in analogy with classical extrapolation methods.
To leverage conjugate computation, here we instantiate the idea using \Acp{GP} in a Bayesian framework.
For the first part, we require an explicit error bound $b : \mathcal{X} \rightarrow [0,\infty)$ such that $b(\mathbf{x}) \geq 0$ with equality if and only if $\mathbf{x} = \mathbf{0}$, and such that $f(\mathbf{x}) - f(\mathbf{0}) = O( b(\mathbf{x}) )$.
The error bound $b$ will be encoded into a centred prior \ac{GP} model for $f$, whose covariance function
\begin{align}
    k(\mathbf{x},\mathbf{x}') := \sigma^2 \left[ k_0^2 + b(\mathbf{x}) b(\mathbf{x}') k_e(\mathbf{x},\mathbf{x}') \right] , \qquad \mathbf{x}, \mathbf{x}' \in \mathcal{X} , \label{eq: cov model}
\end{align}
is selected to ensure samples $g \sim \mathcal{GP}(0,k)$ from the \ac{GP} satisfy, with probability one, $g(\mathbf{x}) - g(\mathbf{0}) = O(b(\mathbf{x}))$ (see \Cref{subsec: sample path} for the precise statement).
Here $\sigma^2 > 0$ is an overall scale to be estimated, while the scalar $k_0^2 > 0$ is proportional to the prior variance of $f(\mathbf{0})$.
The symmetric positive-definite function $k_e : \mathcal{X} \times \mathcal{X} \rightarrow \mathbb{R}$ is the covariance function for the \emph{normalised error} $\mathbf{x} \mapsto e(\mathbf{x})$, where $e(\mathbf{x}) := b(\mathbf{x})^{-1} (f(\mathbf{x}) - f(\mathbf{0}))$ for $\mathbf{x} \in \mathcal{X} \setminus \{\mathbf{0}\}$, and must be specified.
In practice $k_e$ will additionally involve length-scale parameters $\bm{\ell}$ which must be estimated, for example $k_e(\mathbf{x},\mathbf{x}') = \exp(- \sum_{i=1}^d \ell_i^{-2} (x_i - x_i')^2 )$ in the case of the Gaussian kernel; we defer all discussion of this point to \Cref{subsec: estimate order,sec: applications}.
To our knowledge, the encoding of convergence orders into a \ac{GP} as in \eqref{eq: cov model} has not been well-studied, though the basic idea appeared in \citet{tuo2014surrogate} and in our preliminary work \citep{teymur2021black}.
Standard techniques can be applied to fit such a \ac{GP} model to a dataset; see \Cref{fig: illustrating samples} and \Cref{subsec: fitting the GP}.

\begin{remark}[Recovering Richardson in dimension $d = 1$] \label{remark: classical connection}
    Let $k_e$ be any kernel that reproduces the polynomial space $\pi_{n-2}(\mathbb{R})$, such as $k_e(x,x') = (1 + xx')^{n-2}$, and consider the `objective' prior with $k_0^2 \rightarrow \infty$.
    Conditioning on data $\{f(x_i)\}_{i=1}^n$, the posterior mean function is the unique interpolant of the form $x \mapsto \mu + b(x) p(x)$ for some $\mu \in \mathbb{R}$, $p \in \pi_{n-2}(\mathbb{R})$ \citep[see e.g.][Proposition 2.6]{karvonen2018bayes}.
    Thus, if $b$ is polynomial, the intercept $\mu$ is the result of polynomial extrapolation to $0$, and is an instance of Richardson's classical extrapolation method.
\end{remark}

Unfortunately the connection in \Cref{remark: classical connection} is not especially useful.
Indeed, while the posterior mean provides a useful point estimate, the posterior variance is identically zero, meaning that predictive uncertainty is not being properly quantified.
Thus we do not attempt to reproduce Richardson extrapolation in the sequel, but rather we develop \emph{de novo} methodology tailored to the \Ac{GP} framework.

\begin{figure}[t!]
\centering
\includegraphics[width = \textwidth]{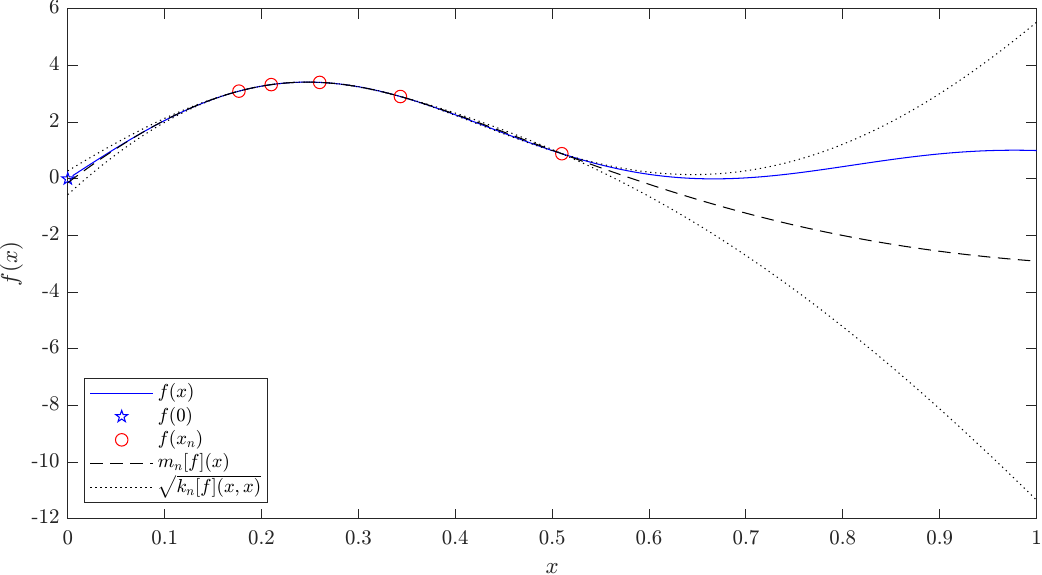}
\caption{The numerical analysis-informed Gaussian process model, fitted to an illustrative dataset $\{f(x_i)\}_{i=1}^n$ (red circles) of size $n =5$, corresponding to the approximations produced by a finite difference method (blue solid curve) whose first-order accuracy [i.e., $b(x) = x$] was encoded into the GP.
The scale $\sigma_n^2[f]$ of the uncertainty was calibrated using the method advocated in \Cref{subsec: UQ}, while $k_e$ was taken to be a Mat\'{e}rn-$\frac{5}{2}$ kernel with length-scale parameter selected using quasi maximum likelihood likelihood (see \Cref{subsec: estimate order}).
Observe that point estimate $m_n[f](0)$ (black dashed curve at $x = 0$), is more accurate than that of the highest fidelity simulation from the numerical method, while the limiting quantity of interest $f(0)$ (blue star) falls within the one standard deviation prediction interval (black dotted curves at $x = 0$).
}
\label{fig: illustrating samples}
\end{figure}

\subsection{Gauss--Richardson Extrapolation}
\label{subsec: fitting the GP}

First we recall the relevant calculations for conditioning the \ac{GP} model \eqref{eq: cov model} on a dataset.
Let $k_b : \mathcal{X} \times \mathcal{X} \rightarrow \mathbb{R}$ be defined as $k_b(\mathbf{x},\mathbf{x}') \coloneqq b(\mathbf{x}) b(\mathbf{x}') k_e(\mathbf{x},\mathbf{x})$, so that our assumptions on $b$ and $k_e$ imply that $k_b$ is a symmetric positive-definite kernel on $\mathcal{X} \setminus \{\mathbf{0}\}$, and a symmetric positive semi-definite kernel on $\mathcal{X}$.
Let $f(X_n)$ be a column vector with entries $f(\mathbf{x}_i)$, let $\mathbf{k}_b(\mathbf{x})$ be a column vector with entries $k_b(\mathbf{x}_i,\mathbf{x})$, and let $\mathbf{K}_b$ be a matrix with entries $k_b(\mathbf{x}_i,\mathbf{x}_j)$.
Recalling that $k_0^2$ is proportional to the prior variance for $f(\mathbf{0})$, we opt for an `objective' prior in which $k_0^2 \rightarrow \infty$.
However, this limit results in an improper prior \ac{GP}.
To make progress, we must first compute the conditional \ac{GP} using a finite value of $k_0^2$ and then retrospectively take the limit -- a standard calculation which we detail in \Cref{app: flat prior limit} -- yielding conditional mean and covariance functions
\begin{align}
m_n[f](\mathbf{x}) & \coloneqq \frac{\mathbf{1}^\top \mathbf{K}_b^{-1} f(X_n) }{ \mathbf{1}^\top \mathbf{K}_b^{-1} \mathbf{1} } + \mathbf{k}_b(\mathbf{x})^\top \mathbf{K}_b^{-1} \left\{ f(X_n) -  \left( \frac{\mathbf{1}^\top \mathbf{K}_b^{-1} f(X_n) }{ \mathbf{1}^\top \mathbf{K}_b^{-1} \mathbf{1} } \right) \mathbf{1} \right\} ,  \label{eq: gp post mean} \\
k_n[f](\mathbf{x},\mathbf{x}') & \coloneqq \sigma_n^2[f] \left\{ k_b(\mathbf{x},\mathbf{x}') - \mathbf{k}_b(\mathbf{x})^\top \mathbf{K}_b^{-1} \mathbf{k}_b(\mathbf{x}') + \frac{ [ \mathbf{k}_b(\mathbf{x})^\top \mathbf{K}_b^{-1} \mathbf{1} - 1 ] [ \mathbf{k}_b(\mathbf{x}')^\top \mathbf{K}_b^{-1} \mathbf{1} - 1 ] }{ \mathbf{1}^\top \mathbf{K}_b^{-1} \mathbf{1} } \right\} ,  \label{eq: gp post cov}
\end{align}
where $\mathbf{1}$ is a column vector whose elements are all 1.
The matrix $\mathbf{K}_b$ can indeed be inverted since we have assumed that the entries of $X_n \subset \mathcal{X} \setminus \{\mathbf{0}\}$ are distinct. 
To obtain \eqref{eq: gp post cov} we have additionally replaced $\sigma^2$ with $\sigma_n^2[f]$, an estimator for the scale parameter $\sigma$, to be specified in \Cref{subsec: UQ}.
Computing the conditional mean and variance at $\mathbf{x} = \mathbf{0}$ results in the simple formulae
\begin{align}
m_n[f](\mathbf{0}) = \frac{ \mathbf{1}^\top \mathbf{K}_b^{-1} f(X_n) }{  \mathbf{1}^\top \mathbf{K}_b^{-1} \mathbf{1} } \qquad \text{ and } \qquad k_n[f](\mathbf{0},\mathbf{0}) =  \frac{ \sigma_n^2[f] }{ \mathbf{1}^\top \mathbf{K}_b^{-1} \mathbf{1} },  \label{eq: simplified expressions}
\end{align}
since $b(\mathbf{0}) = 0$, and thus $k_b(\mathbf{0},\mathbf{x}) = b(\mathbf{0}) b(\mathbf{x}) k_e(\mathbf{0},\mathbf{x}) = 0$ for all $\mathbf{x} \in \mathcal{X}$.
The proposed \ac{GRE} method returns a (univariate) Gaussian distribution, which can be summarised using the point estimate $m_n[f](\mathbf{0})$ for $f(\mathbf{0})$, together with the $100(1-\alpha)\%$ credible intervals
\begin{align}
C_{\alpha}[f] = \left\{ y \in \mathbb{R} : \frac{ | y - m_n[f](\mathbf{0}) | }{ \sqrt{k_n[f](\mathbf{0},\mathbf{0})} } \leq \Phi^{-1}\left(1 - \frac{\alpha}{2} \right) \right\} \label{eq: credible inter}
\end{align}
where $\Phi$ denotes the standard Gaussian cumulative density function.
The uncertainty quantification provided by \ac{GRE} unlocks additional functionality that was not available to classical extrapolation methods, including optimal experimental design for selecting $X_n$ (\Cref{subsec: design}) and principled statistical methods for estimating uncertain convergence orders (\Cref{subsec: estimate order}).
However, both the accuracy of the point estimate and the coverage of the credible intervals will depend critically on the choice of the scale estimator $\sigma_n^2[f]$ and the choice of covariance function $k_e$.
This important issue of how to select $\sigma_n^2[f]$ and $k_e$ will be discussed next.
An illustration of the proposed \ac{GRE} method is provided in \Cref{fig: illustrating samples}.

\subsection{Conservative Gaussian Process Priors}
\label{subsec: justify prior}

Our set-up involves a non-random function $f$ that is modelled using a prior \ac{GP}.
One would perhaps hope to elicit a prior covariance function $k$ in such a manner that $f$ could plausibly have been generated as a sample from the \ac{GP}.
However, such elicitation is fundamentally difficult; the sample support set of a \ac{GP} is not a vector space and may not even be measurable in general \citep{stein2019comment,karvonen2023small}.
How then can we proceed?
In the applications that we have in mind, it is often possible to identify a symmetric positive-definite kernel such that $f$ belongs to the \ac{RKHS} associated to the kernel, whose elements are real-valued functions on $\mathcal{X}$. 
For example, in numerical analysis it is often possible to reason that $f$ possesses a certain number of derivatives, from which inclusion in certain Sobolev \acp{RKHS} can be deduced.
The approach that we take is to identify the covariance function $k$ with the kernel of an \ac{RKHS}, denoted $\mathcal{H}_k(\mathcal{X})$, in which $f$ is contained.
In particular, for any $k_0^2 \in (0,\infty)$ the space reproduced by the kernel $k$ in \eqref{eq: cov model} consists of functions $g : \mathcal{X} \rightarrow \mathbb{R}$ of the form $g(\mathbf{x}) = \mu + b(\mathbf{x}) e(\mathbf{x})$ where $\mu \in \mathbb{R}$ and $e \in \mathcal{H}_{k_e}(\mathcal{X})$, and in the $k_0^2 \rightarrow \infty$ limit the norm structure of $\mathcal{H}_k(\mathcal{X})$ reduces to a semi-norm \smash{$|g|_{\mathcal{H}_k(\mathcal{X})} \coloneqq \| \mathbf{x} \mapsto (g(\mathbf{x}) - g(\mathbf{0}))/b(\mathbf{0}) \|_{\mathcal{H}_{k_e}(\mathcal{X})}$} induced by the norm structure of $\mathcal{H}_{k_e}(\mathcal{X})$; further background on \ac{RKHS} can be found in \citet{berlinet2011reproducing}.
This construction results in a \emph{conservative} prior \ac{GP}, since with probability one sample paths will be less regular than $f$ when the \ac{RKHS} is infinite-dimensional.
However, there are several senses in which this approach to prior elicitation can be justified.
First, it can be viewed as a form of `objective' prior for \acp{GP}, in the sense that it is not intended to reflect prior belief but is rather intended to induce desirable behaviour in the posterior \ac{GP}.
Second, the choices that we make here will be justified through theoretical guarantees on both point estimation error (\Cref{subsec: faster convergence}) and coverage of credible sets (\Cref{subsec: UQ}).
Third, the introduction of an additional scale estimator $\sigma_n^2[f]$ in \eqref{eq: gp post cov} provides an opportunity to counteract the conservatism of the choice of $k$ through the data-driven estimation of an appropriate scale for the credible sets in \eqref{eq: credible inter}.

\subsection{Higher-Order Convergence Guarantees}
\label{subsec: faster convergence}

The main technical contribution of this paper is to establish sufficient conditions under which the \ac{GRE} point estimate $m_n[f](\mathbf{0})$ in \eqref{eq: simplified expressions} provides a more accurate approximation to the continuum limit $f(\mathbf{0})$ compared to the highest fidelity approximation $f(\mathbf{x}_n)$ on which it is based.
The analysis we present is based on local polynomial reproduction, similar to that described in \citet{wendland2004scattered}.
However, our results differ from existing work in that they are adapted to the non-stationary kernel \eqref{eq: cov model} and quantify the space-filling properties of a design $X_n$ using boxes, rather than balls or cones, since boxes are more natural for the domain $\mathcal{X} \subseteq [0,\infty)^d$ and enable sharper control over the constants involved.

To state our results, we define the \emph{box fill distance} $\rho_{X_n,\mathcal{X}}$ as the supremum value of $\nu$ such that there is a box of the form $[\mathbf{x} , \mathbf{x} + \nu \mathbf{1}]$ contained in $\mathcal{X}$ for which $X_n \cap [\mathbf{x} , \mathbf{x} + \nu \mathbf{1}] = \emptyset$.
Define the constants $\gamma_d$ using the induction $\gamma_d \coloneqq 2d(1+\gamma_{d-1})$ with base case $\gamma_1 \coloneqq 2$.
Our first main result, whose proof is contained in \Cref{app: proof thm 1}, concerns the finite-smoothness case where polynomial-order acceleration can be achieved:

\begin{theorem}[Higher-order convergence; finite smoothness]
\label{prop: mean conver}
Let $\mathcal{X} = [\mathbf{0},\mathbf{1}] \subset \mathbb{R}^d$ and $X_n \subset \mathcal{X}$.
Let $\mathcal{X}_h = [\mathbf{0},h\mathbf{1}]$ and $X_n^h = \{h \mathbf{x} : \mathbf{x} \in X_n\}$ where $h \in (0,1]$.
Assume that $f \in \mathcal{H}_k(\mathcal{X})$, $b \in \pi_r(\mathcal{X})$ and $k_e \in C^{2s}(\mathcal{X} \times \mathcal{X})$.
Let $m_n^h[f](\mathbf{0})$ denote the point estimate \eqref{eq: simplified expressions} based on data $f(X_n^h)$.
Then there is an explicit $n$- and $h$-independent constant $C_{r,s}$, defined in the proof, such that
\begin{align*}
\underbrace{ |f(\mathbf{0}) - m_n^h[f](\mathbf{0})| }_{\text{\normalfont extrapolation error}} \; \leq \; C_{r,s} \rho_{X_n,\mathcal{X}}^s |f|_{\mathcal{H}_k(\mathcal{X})} \; \underbrace{ \phantom{|} h^s}_{ \text{\normalfont acceleration}} \; \underbrace{ \|b\|_{L^\infty(\mathcal{X}_h)} }_{\text{\normalfont original bound }}
\end{align*}
whenever the box fill distance satisfies $\rho_{X_n,\mathcal{X}} \leq 1 / (\gamma_d (r + 2s))$.
\end{theorem}

\noindent To interpret the conclusion of \Cref{prop: mean conver}, fix $n$ to be large enough that the constraint on the box fill distance is satisfied and examine the convergence of $m_n^h[f](\mathbf{0})$ to $f(\mathbf{0})$ as $h$ is decreased.
If the problem possesses no additional smoothness to exploit (i.e., $s=0$) then convergence is gated at the rate $\|b\|_{L^\infty(\mathcal{X}_h)}$ of the original numerical method, irrespective of the number $n$ of data that are used to train the \ac{GP}.
On the other hand, if $f$ is regular enough that the normalised error functional $\mathbf{x} \mapsto (f(\mathbf{x}) - f(\mathbf{0})) / b(\mathbf{x})$ is an element of the \ac{RKHS} $\mathcal{H}_{k_e}(\mathcal{X})$ of an $s$-smooth kernel (implied by $|f|_{\mathcal{H}_k(\mathcal{X})} < \infty$), then the $h^s$ factor provides acceleration of polynomial order $s$ over the convergence rate of the original numerical method.
To the best of our knowledge these theoretical results are the first of their kind for convergence acceleration using \acp{GP}.
\Cref{ex: finite diff,ex: Romberg} illustrate cases in which our regularity assumptions are satisfied.
For the reader's convenience, we recall some standard examples of kernels and their associated smoothness properties in \Cref{app: kernels}.

\begin{remark}[Sample efficiency compared to Richardson]
A notable feature of Richardson extrapolation is that, under appropriate regularity assumptions, acceleration of order $s$ can be achieved using a dataset of size $n = s + 1$ in dimension $d = 1$.
For example, if $f$ is first-order accurate with $f(h) = f(0) + c_1 h + O(h^2)$, then the line that passes through data $(h,f(h))$ and $(2h,f(2h))$ has intercept $2f(h) - f(2h)$, which is equal to $2[f(0) + c_1 h + O(h^2)] - [f(0) + 2c_1 h + O(h^2)] = f(0) + O(h^2)$; an additional order of accuracy is gained. 
Our result is less sample-efficient, in the sense that $n \geq 2r + 4s$ data are in principle required, due to the constraint on the box fill distance in \Cref{prop: mean conver}.
However, we speculate that this lower bound on $n$ is not tight, and we empirically confirm that order-$s$ acceleration is observed at smaller sample sizes $n$ in \Cref{ex: finite diff,ex: Romberg}.
\end{remark}

On the other hand, if there is infinite smoothness to exploit, then we may consider increasing the value of $s$ in \Cref{prop: mean conver} to obtain an arbitrarily fast convergence rate as $h \rightarrow 0$, albeit with an increasing number $n$ of training points required for the bound to hold.
This result goes beyond classical Richardson extrapolation, but is natural within the \ac{GP} framework.
\Cref{cor: spectral}, whose proof is contained in \Cref{app: proof of spectral}, is obtained by carefully tracking the $s$-dependent constants appearing in \Cref{prop: mean conver}:

\begin{theorem}[Higher-order convergence; infinite smoothness] \label{cor: spectral}
In the setting of \Cref{prop: mean conver}, assume further that $k_e \in C^{\infty}(\mathcal{X} \times \mathcal{X})$ and that $\sup_{\mathbf{x},\mathbf{y} \in \mathcal{X}} \sum_{|\bm{\beta}| = 2s} |\partial_{\mathbf{y}}^{\bm{\beta}} k_e(\mathbf{x},\mathbf{y})| \leq C_k^{2s} (2s)!$ for some constant $C_k$.
Then there exists an explicit $h$-independent constant $C_{n,r,s}$, defined in the proof, such that
\begin{align*}
\underbrace{ |f(\mathbf{0}) - m_n^h[f](\mathbf{0})| }_{\text{\normalfont extrapolation error}} \; \leq \; C_{n,r,s} |f|_{\mathcal{H}(k)} \; \underbrace{ 
\phantom{|}h^{ \scaleto{ \frac{1}{ 4 \gamma_d \rho_{X_n,\mathcal{X}} } }{15pt} } }_{\text{\normalfont acceleration}} \; \underbrace{ \|b\|_{L^\infty(\mathcal{X}_h)} }_{\text{\normalfont original bound }}
\end{align*} 
whenever the box fill distance satisfies $\rho_{X_n,\mathcal{X}} \leq \min \{ 1 / (2 \gamma_d (r+1)) , 1 / (2 d^{1/2} \gamma_d e^{4 d \gamma_d + 1} ) \}$.
\end{theorem}

\noindent The derivative growth condition in the statement of \Cref{cor: spectral} holds for most popular smooth kernels $k_e$, including the Gaussian kernel.
The order of acceleration is now determined by the box fill distance, which reflects the general phenomenon that ``more samples are required to exploit smoothness'' \citep{cabannes2023many}.

To assess the sharpness of our results we first consider the problem of approximating derivatives using finite differences; a setting where extrapolation methods are routinely used \citep[see Section~6.7 of][]{brezinski2013extrapolation}:

\begin{example}[Higher-order convergence for finite difference approximation]
\label{ex: finite diff}
Consider numerical differentiation of a suitably regular function $\psi : \mathbb{R} \rightarrow \mathbb{R}$.
The \emph{central difference method}
$$
f(x) \coloneqq \frac{ \psi(t+x) - \psi(t-x) }{ 2x } , \qquad x > 0,
$$
is a second-order approximation to $\psi'(t)$ for a given $t \in \mathbb{R}$.
To make use of our results we set $b(x) = x^2$, from \eqref{eq: cov model}, and suppose that $\psi(t+x) = c_0 + c_1 x + c_2 x^2 + c_3(x) x^3 $ for some $c_0,c_1,c_2 \in \mathbb{R}$ and some $x$-dependent coefficient $c_3(x)$.
The normalised error is
$$
e(x) = \frac{ f(x) - f(0) }{ b(x) } 
= \frac{ \psi(t+x) - \psi(t - x) }{ 2 x \cdot x^2 } - \frac{ \psi'(t) }{ x^2 }
= \frac{ c_3(x)-c_3(-x) }{ 2 } , 
$$
so that the assumptions of \Cref{prop: mean conver} are satisfied when $x \mapsto c_3(x)$ and $x \mapsto c_3(-x)$ are elements of $\mathcal{H}_{k_e}(\mathcal{X})$, and $k_e \in C^{2s}(\mathcal{X} \times \mathcal{X})$
(the latter condition can be satisfied for example by taking $k_e$ to be either a Mat\'{e}rn kernel or a Wendland kernel with appropriate smoothness level; see \Cref{app: kernels}).
As a test problem, consider $\psi(t) = \sin(10t) + 1_{t > 0} t^{s+4}$, with $\psi'(0) = 10$ the value to be estimated; in \Cref{app: check assum} we verify that our assumption on the form of $\psi$ is satisfied.
The sample size $n = 5$ was fixed and the initial design $X_n = \{0.2,0.4,0.6,0.8,1\}$ was scaled by a factor $h$ to obtain a range of designs $X_n^h \subset (\mathbf{0},h\mathbf{1}]$.
In these experiments we work in 100 digits of numerical precision, so that rounding error can be neglected.

Results for $s=2$ are reported in \Cref{fig: convergence rates}, with the \emph{absolute error} $|f(0) - m_n^h[f](0)|$ plotted as a function of $h$ in the left panel.
These results reveal that the orders of acceleration predicted by our analysis are achieved, despite the sample size $n$ being less than that required to fulfill the box fill distance requirement in \Cref{prop: mean conver}.
The \ac{GRE} method demonstrated accuracy comparable to Richardson's extrapolation method (and superior to other classical extrapolation methods) when the kernel was chosen to match the smoothness of the task at hand.
Interestingly, the most accurate extrapolation was provided by \ac{GRE} with the Gaussian kernel, despite this kernel being too smooth for the task at hand.
The coverage of \ac{GRE} credible intervals was also investigated, with the \emph{relative error}
$(f(0) - m_n^h[f](0)) / \sqrt{k_n[f](0,0)}$ 
plotted as a function of $h$ in the right panel.
It was found that credible intervals are asymptotically conservative in the case where a kernel with finite smoothness was used, in the sense that the relative error appeared to vanish in the $h \rightarrow 0$ limit.
However, in the case of the Gaussian kernel the credible intervals appeared to be asymptotically calibrated, in the sense that the relative error appeared to converge to a finite value ($\approx 3$) in the $h \rightarrow 0$ limit.
Theoretical analysis of the \ac{GRE} credible intervals is provided in \Cref{subsec: UQ}.
\end{example}

\begin{figure}[t!]
    \centering
    \includegraphics[width = 0.5\textwidth]{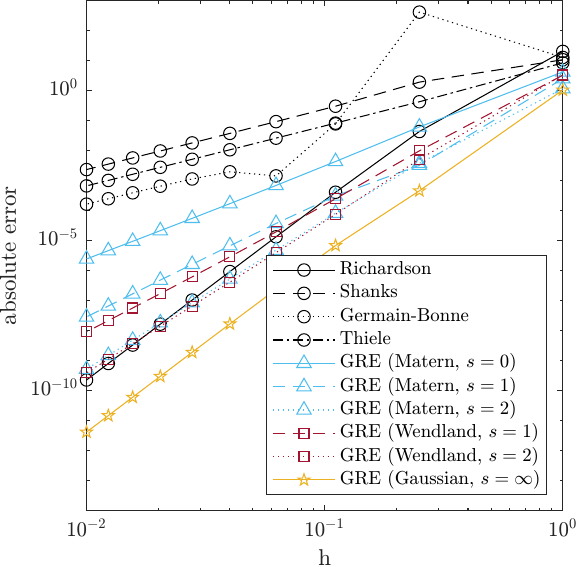}
    \includegraphics[width = 0.48\textwidth]{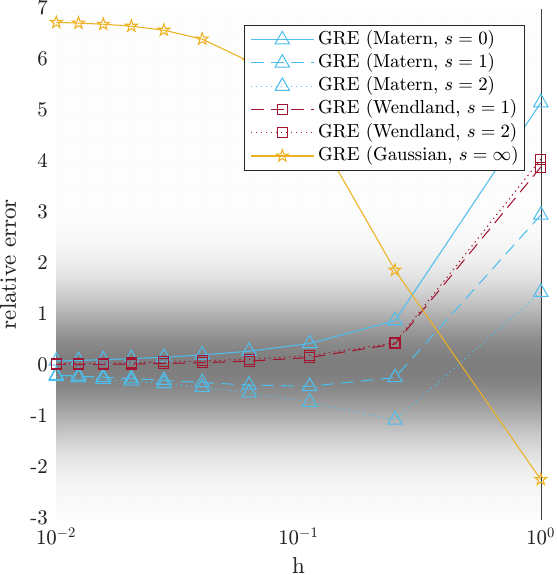}
    \caption{Accelerating the central difference method; \Cref{ex: finite diff}.
    The left panel presents the absolute error $|f(0) - m_n^h[f](0)|$, while the right panel presents the relative error $(f(0) - m_n^h[f](0)) / \sqrt{k_n[f](0,0)}$.
    Classical extrapolations methods (black circles) were compared to our Gauss--Richardson Extrapolation (GRE) method, with either a Mat\'{e}rn (blue triangles), Wendland (red squares), or Gaussian (yellow stars) kernel.
    The true smoothness in this case is $s = 2$, while the legend indicates the level of smoothness assumed by the kernel.
    Kernel length-scale parameters were set to $\ell = 1$ and the scale estimator $\sigma_n^2[f]$ proposed in \Cref{subsec: UQ} was used.
    Shaded regions in the right panel correspond to the density function of the standard normal.
    }
    \label{fig: convergence rates}
\end{figure}

Though they accurately describe the convergence acceleration provided by the \ac{GRE} method, there are at least two apparent drawbacks with \Cref{prop: mean conver,cor: spectral}.
The first is that these results require the error bound $b$ to be a polynomial; this is an intrinsic part of our proof strategy, which is based on local polynomial reproduction, and cannot easily be relaxed. 
However, for applications in which a non-polynomial error bound $\tilde{b}$ naturally arises, we may still be able to construct a polynomial error bound $b \in \pi_r(\mathcal{X})$ for some $r$ that satisfies $\tilde{b}(\mathbf{x}) \leq b(\mathbf{x})$ for all $\mathbf{x} \in \mathcal{X}$ and enables the conclusion of \Cref{prop: mean conver} to be applied.
The second limitation is that, for many iterative numerical methods that produce a convergent sequence of approximations to the continuum quantity $f(\mathbf{0})$ of interest, there is not always the notion of a continuum of discretisation parameters $\mathbf{x}$ that can be exploited in the \ac{GRE} framework.
This second issue can be elegantly addressed using the notion of an $s$-smooth extension, which we introduce next.

\subsection{The Generality of Continua}
\label{subsec: continuous and discrete}

Several numerical methods do not admit a continuum of discretisation parameters $\mathbf{x}$ that can be exploited in the \ac{GRE} method.
For example, the conjugate gradient algorithm for approximating the solution to a linear system of equations produces a convergent sequence of approximations, but is in no sense continuously indexed.
The aim of this section is to demonstrate that iterative methods, which produce a sequence of approximations converging to a limiting quantity of interest, do in fact fall within our framework.
The idea, roughly speaking, is to construct a function $f$ whose values $f(x_n)$ on a convergent sequence, such as $x_n = 1/n$, coincide with with the approximation produced after $n$ iterations of the numerical method.
The challenge is to show that such a function $f$ exists with sufficient regularity that the results of \Cref{subsec: faster convergence} can be applied.
Our main tool is the idea of an $p$-smooth extension, which is the content of \Cref{prop: discrete}.
Let $\min(\mathbf{z}) \coloneqq \min\{z_1,\dots,z_d\}$ for $\mathbf{z} \in \mathbb{R}^d$.

\begin{proposition}[$p$-smooth extension] \label{prop: discrete}
Suppose that $C^p(\mathcal{X}) \subset \mathcal{H}_{k_e}(\mathcal{X})$ for some $p \in \mathbb{N}$.
Let $(\mathbf{x}_n)_{n \in \mathbb{N}} \subset \mathcal{X} \setminus \{\mathbf{0}\}$ be such that $\mathbf{x}_{n+1} < \mathbf{x}_n$ componentwise and $\mathbf{x}_n \rightarrow \mathbf{0}$.
Let $(y_n)_{n \in \mathbb{N}}$ be a convergent sequence with limit $y_\infty$, such that the normalised errors $e_n \coloneqq (y_n - y_\infty) / b(\mathbf{x}_n)$ satisfy $|e_n - e_{n+1}| / \min (\mathbf{x}_n - \mathbf{x}_{n+1})^p \rightarrow 0$.
Then there exists a function $f$ such that $f(\mathbf{0}) = y_\infty$, $f(\mathbf{x}_n) = y_n$ for each $n \in \mathbb{N}$, and $|f|_{\mathcal{H}_k(\mathcal{X})} < \infty$.
\end{proposition}

\noindent A polynomial expansion can be used to establish the preconditions of \Cref{prop: discrete}, as we illustrate in the following result:

\begin{corollary}[Sufficient conditions for $p$-smooth extension in $d = 1$] \label{cor: second order expand}
    Let $(x_n,y_n)_{n \in \mathbb{N}} \subset (0,\infty) \times \mathbb{R}$ be such that $x_n$ converges monotonically to 0, with 
    $(x_n^{p+1} - x_{n+1}^{p+1}) (x_n - x_{n+1})^{-p} \rightarrow 0$, $x_n^{p+2}(x_n - x_{n-1})^{-p} \rightarrow 0$
    and $y_n = y_\infty + C_1 x_n^{r} + C_2 x_n^{r + p + 1} + O(x_n^{r + p + 2})$ for some constants $y_\infty, C_1, C_2 \in \mathbb{R}$.
    Let $b(x) = x^{r}$.
    Then the preconditions of \Cref{prop: discrete} are satisfied.
\end{corollary}

\noindent The proof of both \Cref{prop: discrete} and \Cref{cor: second order expand} can be found in \Cref{subsec: proof of continuous extend results}.
The conditions on the sequence $(x_n)_{n \in \mathbb{N}}$ in \Cref{cor: second order expand} are satisfied by, for example, sequences of the form $x_n = \frac{1}{n}$ and $x_n = \lambda^{-n}$ for any $\lambda > 1$, which are the sort of expressions that routinely appear in error bounds. 
The overall approach is illustrated in \Cref{ex: Romberg}, where a GP analogue of the classical Romberg method for numerical integration is derived.

\begin{example}[GP Romberg methods] \label{ex: Romberg}
Romberg methods for numerical integration are classically obtained via Richardson extrapolation of the trapezoidal rule \citep[][Section~6.7]{brezinski2013extrapolation}; it is interesting to ask if a similar feat can be achieved with \ac{GRE}.
Let $\psi \in C^{2m+2}([0,1])$ and consider the trapezoidal rule
$
y_n \coloneqq \frac{1}{n} [ \frac{\psi(0)}{2} + \psi\left( \frac{1}{n} \right) + \dots + \psi\left( \frac{n-1}{n} \right) + \frac{\psi(1)}{2} ] 
$.
The Euler--Maclaurin summation formula implies that the error of the trapezoidal rule can be expressed as 
$$
y_n - \int_0^1 \psi(t) \; \mathrm{d}t = \sum_{i=1}^m \frac{B_{2i}}{(2i)!} x_n^{2i} \left( \psi^{(2i-1)}(1) - \psi^{(2i-1)}(0) \right) + \frac{B_{2m+2}}{(2m+2)!}  x_n^{2m+2} \psi^{(2m+2)}(\beta_n)
$$
for some $\beta_n \in [0,1]$, where $x_n = \frac{1}{n}$ and $B_k$ are the Bernoulli numbers.
As a test problem, consider $\psi(t) = \sin(10t) + t^2$, for which we can apply \Cref{cor: second order expand} with $b(x) = x^2$, $r = 2$ and $p = 3$.
Thus there exists a function $f$ that agrees with the trapezoidal rule on $(x_n)_{n \in \mathbb{N}}$ and satisfies the preconditions of \Cref{prop: mean conver} for a kernel $k_e$ with smoothness up to $s = 2$; see \Cref{app: kernels}.
Empirical results in \Cref{fig: convergence rates trapz} verify that we are indeed able to gain an additional $s=2$ convergence orders over the original trapezoidal rule, akin to Romberg integration, using our \ac{GRE} method.
Here the sample size $n = 5$ was fixed and the initial design $X_n = \{1, \frac{1}{2}, \frac{1}{4}, \frac{1}{8}, \frac{1}{16}\}$ was scaled by a factor $h$ to obtain a range of designs $X_n^h \subset (\mathbf{0},h\mathbf{1}]$.
The accuracy of the \ac{GRE} point estimator and the coverage of the \ac{GRE} credible interval demonstrate similar behaviour to that observed in \Cref{ex: finite diff}.
\end{example}

\begin{figure}[t!]
    \centering
    \includegraphics[width = 0.5\textwidth]{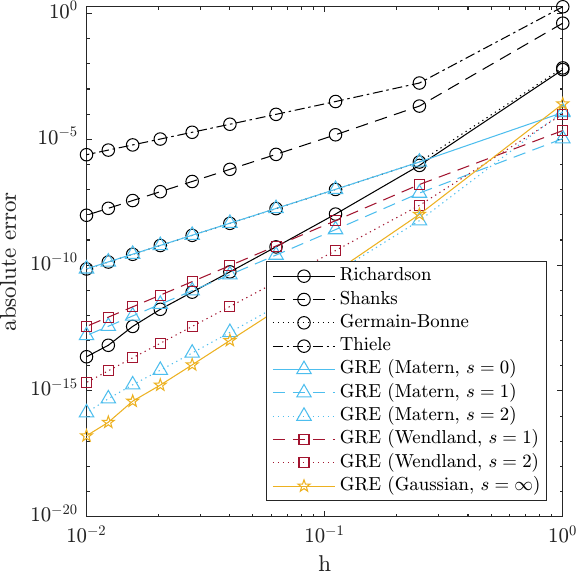}
    \includegraphics[width = 0.49\textwidth]{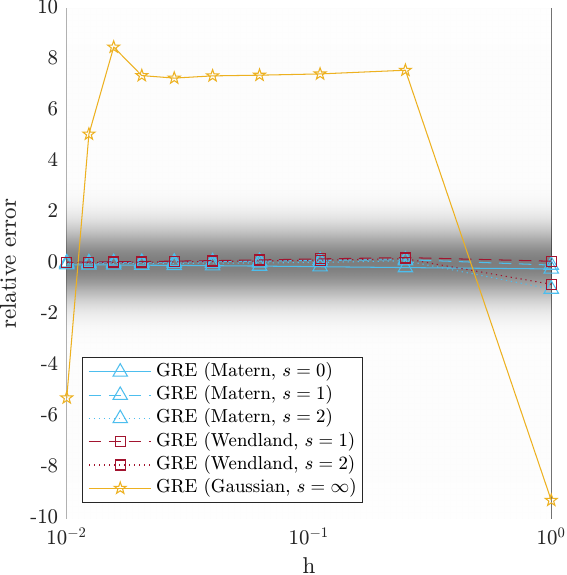}
    \caption{Accelerating the trapezoidal method to obtain a GP Romberg method; \Cref{ex: Romberg}.
    The left panel presents the absolute error $|f(0) - m_n^h[f](0)|$, while the right panel presents the relative error $(f(0) - m_n^h[f](0)) / \sqrt{k_n[f](0,0)}$.
    Classical extrapolations methods (black circles) were compared to our Gauss--Richardson Extrapolation (GRE) method, with either a Mat\'{e}rn (blue triangle), Wendland (red squares), or Gaussian (yellow stars) kernel.
    The true smoothness in this case is $s = 2$, while the legend indicates the level of smoothness assumed by the kernel.
    Kernel length-scale parameters were set to $\ell = 1$.
    Shaded regions in the right panel correspond to the density function of the standard normal.}
    \label{fig: convergence rates trapz}
\end{figure}

These results extend the applicability of \ac{GRE} to settings where components of the discretisation parameter vector $\mathbf{x}$ could take values in any infinite set.
For example, in standard implementations of the finite element method for numerically solving partial differential equations one has a continuous parameter, characterising the width of a triangular mesh, a discrete parameter, characterising the number of cubature nodes used to integrate against each element, and another discrete parameter, specifying the number of iterations of a conjugate gradient method to solve the resulting linear system.
The resulting mixture of continuous and discrete discretisation parameters $\mathbf{x}$ falls within the scope of our \ac{GRE} method.

\subsection{Uncertainty Quantification}
\label{subsec: UQ}

An encouraging observation from \Cref{ex: finite diff,ex: Romberg} was that the \ac{GRE} credible intervals were not asymptotically over-confident as $h \rightarrow 0$.
The aim of this section is to explain how the scale parameter $\sigma^2$ in \eqref{eq: cov model}, which controls the size of credible intervals $C_\alpha[f]$ in \eqref{eq: credible inter}, was actually estimated, and to rigorously prove that asymptotic over-confidence cannot occur when our proposed estimator $\sigma_n^2[f]$ is used.

The most standard approach to kernel parameter estimation is maximum (marginal) likelihood, but in \ac{GRE} we do not have a valid likelihood due to taking the improper $k_0^2 \rightarrow \infty$ limit.
Instead, we motivate a particular estimator $\sigma_n^2[f]$ using asymptotic guarantees for the associated credible interval.
Specifically, we advocate the estimator
\begin{align} 
\sigma_n^2[f] 
\coloneqq \frac{| m_n[f] |_{\mathcal{H}_k(\mathcal{X})}^2 }{n}
= \frac{1}{n} \left[ f(X_n)^\top \mathbf{K}_b^{-1} f(X_n) -  \frac{ ( \mathbf{1}^\top \mathbf{K}_b^{-1} f(X_n) )^2 }{ \mathbf{1}^\top \mathbf{K}_b^{-1} \mathbf{1} } \right] , \label{eq: sigma estimator}
\end{align}
which takes the same form as the maximum likelihood estimator that we would have obtained had we not taken the $k_0^2 \rightarrow \infty$ limit, but with the semi-norm $| m_n[f] |_{\mathcal{H}_k(\mathcal{X})}$ in place of the conventional norm on $\mathcal{H}_k(\mathcal{X})$.
This choice is supported by the following asymptotic result, whose proof is contained in \Cref{app: uq proof}:

\begin{proposition}[Asymptotic over-confidence is prevented] \label{prop: not over-confident}
  In the setting of \Cref{prop: mean conver}, suppose that $s \geq 1$ and that $\lim_{\mathbf{x} \rightarrow \mathbf{0}} b(\mathbf{x})^{-1} (f(\mathbf{x}) - f(\mathbf{0})) \neq 0$ (i.e. we have a sharp error bound).
  Let $m_n^h[f](\mathbf{0})$ and $k_n^h[f](\mathbf{0}, \mathbf{0})$ denote the conditional mean and variance in~\eqref{eq: simplified expressions}, based on data $f(X_n^h)$ and the estimator in~\eqref{eq: sigma estimator}. 
  Then
  \begin{equation*}
    \limsup_{h \to 0} \frac{ \lvert f(\mathbf{0}) - m_n^h[f](\mathbf{0}) \rvert }{  \sqrt{ k_n^h[f](\mathbf{0},\mathbf{0})} } < \infty
  \end{equation*}
  whenever the box fill distance $\rho_{X_n, \mathcal{X}}$ is sufficiently small.
\end{proposition}

\noindent In other words, the width $\sqrt{k_n[f](\mathbf{0},\mathbf{0})}$ of the credible interval cannot vanish asymptotically faster than the actual absolute error $|f(\mathbf{0}) - m_n[f](\mathbf{0})|$.
Though this result does not guarantee that credible intervals are the `right size' \textit{per se}, there is no randomness in the data-generating process $f(\mathbf{x})$ and thus standard statistical notions of coverage, or `right size', cannot be directly applied \citep[see][]{karvonen2020maximum}.
In practice, we have already seen empirical evidence that the credible sets \eqref{eq: credible inter} are appropriately conservative; an arguably predictable consequence of the conservative \ac{GP} prior discussed in \Cref{subsec: justify prior}.
Note that the conclusion of \Cref{prop: not over-confident} also holds when the stronger hypotheses of \Cref{cor: spectral} are assumed.
However, the result assumes that a kernel with appropriate smoothness is used; it does not explain the behaviour of \ac{GRE} with the Gaussian kernel observed in \Cref{ex: finite diff,ex: Romberg}, since in that case the Gaussian kernel was formally misspecified.

Assured that our credible intervals are in a sense meaningfully related to the actual error, we can now proceed to exploit this measure of uncertainty for experimental design.

\subsection{Optimal Experimental Design}
\label{subsec: design}

One of the main engineering challenges associated with the simulation of continuum mathematical or physical phenomena is the numerical challenge of simultaneously controlling all sources of discretisation error, to ensure the output $f(\mathbf{x})$ remains close in some sense to $f(\mathbf{0})$, the continuum quantity of interest.
In practice, one might explore the sensitivity of the simulator output $f(\mathbf{x})$ to small changes in each discretisation parameter $x_i$ in turn, to heuristically identify a global setting $\mathbf{x}_{\text{hi-fi}}$ which is then fixed for the lifetime in which the simulator is used.
It seems remarkable that more principled methodology has not yet been developed, and we aim to fill this gap by formulating \emph{optimal experimental design} within the \ac{GRE} framework.

The accuracy of the point estimator \eqref{eq: simplified expressions} will depend crucially on the locations at which the \ac{GP} has been trained.
\Cref{subsec: UQ} established that the conditional variance is meaningfully related to estimation accuracy, with the advantage that it can be explicitly calculated.
This motivates the following cost-constrained optimisation problem
\begin{align}
    \argmax_{X \subset \mathcal{D}} \mathbf{1}^\top \mathbf{K}_b^{-1} \mathbf{1}
    \qquad \text{s.t.} \qquad  
    \sum_{\mathbf{x} \in X} c(\mathbf{x}) \leq C  , \label{eq: OED2}
\end{align}
where $\mathcal{D} \subseteq \mathcal{X}$ denotes the set of feasible simulations being considered, $\mathbf{K}_b$ is the matrix with entries $k_b(\mathbf{x}_i,\mathbf{x}_j)$, $\mathbf{x}_i,\mathbf{x}_j \in X$, the map $c : \mathcal{D} \rightarrow \mathbb{R}$ quantifies the cost associated with obtaining simulator output $f(\mathbf{x})$, and $C$ denotes the total computational budget.
This numerical analysis-informed objective $\mathbf{1}^\top \mathbf{K}_b^{-1} \mathbf{1}$ is inversely proportional to the \ac{GRE} posterior variance \eqref{eq: simplified expressions} when the scale parameter $\sigma$ is fixed, rather than estimated (since \emph{a priori} we do not suppose data have been obtained from which $\sigma$ could be estimated).
This optimisation does not enforce a particular training sample size $n$, it just constrains the total computational cost.
As such, \eqref{eq: OED2} represents a challenging optimisation problem, with both the number $n$ of experiments in the optimal design, and the optimal experiments $X = \{\mathbf{x}_i\}_{i=1}^n$ themselves, to be determined.
To proceed, we consider a finite set $\mathcal{D}$ of candidate experiments and then use brute force to search for an optimal design restricted to this candidate set.

\begin{example}[Optimal experimental design in $d=1$]
\label{ex: design 2}
    Consider a first-order numerical method with linear cost, so that $b(x) = x$ and $c(x) = x^{-1}$, an example of which would be the classical forward Euler method.
    For illustration we take $k_e$ to be either a Mat\'{e}rn kernel ($s = 0$) or the Gaussian kernel ($s = \infty$), in each case with length-scale $\ell = 1$ fixed.
    The total computational budget $C$ was varied and optimal designs $X$ were computed with elements constrained to a size 20 grid $\mathcal{D}$; results are shown in \Cref{fig: illustrating design 2}.
    In the case of a rough kernel, like the Mat\'{e}rn kernel, a greedy/exploitative strategy of assigning all compute power to the highest resolution experiment seems optimal.
    Since we are working only with a discrete set of experiments, there is a small residual computational budget that is allocated to one or two further cheap experiments.
    For the Gaussian kernel, the optimal strategy is less greedy, with optimal designs involving more experiments, indicating that the greater smoothness is being leveraged to improve the accuracy of \ac{GRE}.
\end{example}

\begin{figure}[t!]
\centering
\includegraphics[width = 0.49\textwidth]{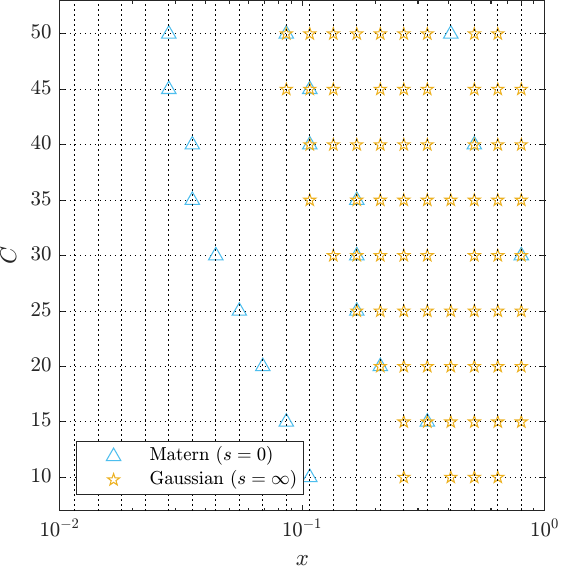}
\includegraphics[width = 0.5\textwidth,clip,trim = 0.7cm 0.6cm 0.8cm 0.7cm]{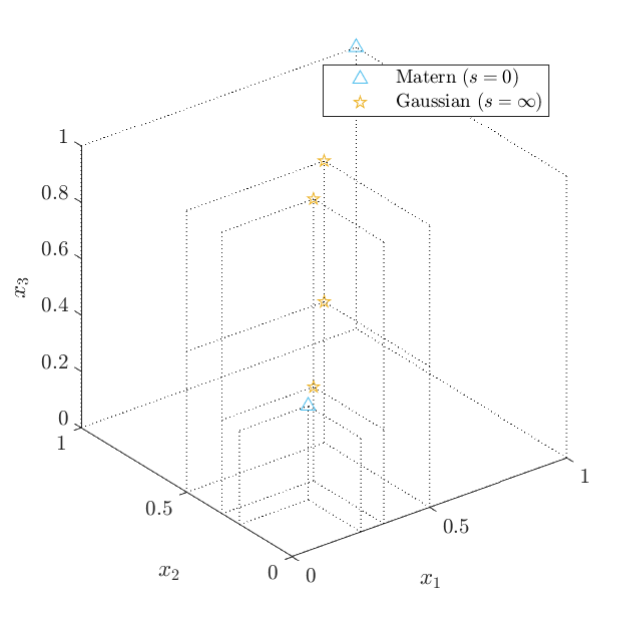}
\caption{Optimal experimental designs were computed, for varying total computational budgets $C$, using either a Mat\'{e}rn (blue triangles; $s = 0$) or Gaussian (yellow stars; $s = \infty$) kernel.
Left: The setting of \Cref{ex: design 2}, with candidate states shown as vertical dotted lines on the plot.
Right:  An illustration of experimental design in dimension $d = 3$, with dotted lines used to indicates the coordinates of the states that were selected.
}
\label{fig: illustrating design 2}
\end{figure}

In practice a small number of preliminary simulations should be used to estimate appropriate length-scale parameters $\bm{\ell}$ for the covariance kernel.
Such parameter estimation becomes more critical in the multivariate setting, illustrated in the right panel of \Cref{fig: illustrating design 2}, since the simulator output $f(\mathbf{x})$ may be more or less sensitive to different components of $\mathbf{x}$; in \Cref{sec: applications} a practical workflow is presented.

\begin{remark}[Trivial solution for iterative methods]
In \Cref{subsec: continuous and discrete} we discussed the scenario where data are generated along a sequence $(\mathbf{x}_n)_{n \in \mathbb{N}}$ by an iterative method, which first produces $f(\mathbf{x}_1), \dots, f(\mathbf{x}_{n-1})$ en route to producing the final output $f(\mathbf{x}_n)$.
In this scenario the cost of computing $f(\mathbf{x}_1),\dots,f(\mathbf{x}_n)$ is simply $c(\mathbf{x}_n)$, in which case computing as many iterations as possible is optimal in the sense of \eqref{eq: OED2}.
\end{remark}

The methodology just presented systematises the often \emph{ad-hoc} process of selecting appropriate fidelities on which simulator output is computed, in a manner that is specifically tailored to improving the accuracy of our \ac{GRE} method.
Sequential experimental design strategies can also be developed, but were not pursued.
The remainder of this section deals with three important generalisations of the \ac{GRE} method; the case where convergence orders are unknown and must be estimated (\Cref{subsec: estimate order}), the case of multivariate simulator output (\Cref{subsec: multivar output}), and the case where the simulator contains additional degrees of freedom (\Cref{subsec: extra parameters}).

\subsection{Extension to Unknown Convergence Order}
\label{subsec: estimate order}

The practical application of extrapolation methods does not necessarily require access to an explicit error bound, as several procedures have been developed to automatically identify a suitable method from a collection of extrapolation methods (which could correspond to different assumed convergence orders, or different classes of extrapolation method).
A representative approach, called \emph{automatic selection} \citep{delahaye1981automatic}, is based on the idea that small changes $S_{n+1} - S_n$ between consecutive iterates is a useful proxy for the convergence rate of an extrapolation method $(S_n)_{n \in \mathbb{N}}$.
Another approach is to linearly combine estimates produced by a collection of extrapolation methods, called a \emph{composite sequence approach} \citep{brezinski1985composite}.
From our statistical standpoint, these methods bear a respective semblance to model selection and model averaging.
Pursuing a statistical perspective on extrapolation, here we consider maximum (marginal) likelihood as a default for selecting an appropriate \ac{GP} prior model for \ac{GRE}.
The $k_0^2 \rightarrow \infty$ limit taken in \Cref{subsec: fitting the GP} means that we do not have a proper likelihood, so instead we identify and maximise an appropriate \emph{quasi} likelihood.
Our justification is twofold, namely (1) our quasi likelihood is directly analogous to the standard \ac{GP} likelihood, and (2) we provide analysis below that demonstrates the consistency of maximum quasi likelihood for estimation of convergence order in the \ac{GRE} framework.

To formulate the main result of this section, we suppose we have a vector $\mathbf{r} \in \mathbb{R}^p$ that parametrises the error bound $b_{\mathbf{r}} : \mathcal{X} \rightarrow [0,\infty)$, with the interpretation that increasing the value of any of the components of $\mathbf{r}$ corresponds to faster convergence of the error bound $b_{\mathbf{r}}(\mathbf{x})$ to $\mathbf{0}$ as $\mathbf{x} \rightarrow \mathbf{0}$.
Specifically, we call a class of error bounds \emph{monotonically parametrised} if, for all $\mathbf{r}_1 < \mathbf{r}_2$ we have
$$
\inf_{\mathbf{r} \leq \mathbf{r}_1} \lim_{\mathbf{x} \rightarrow \mathbf{0}} \frac{b_{\mathbf{r}_2}(\mathbf{x})}{b_{\mathbf{r}}(\mathbf{x})} = 0 .
$$
This is not a restriction \emph{per se}, as we are free to choose how $b_{\mathbf{r}}$ is parametrised, but an assumption of this kind is required to enable the following result to be rigorously stated.
Examples of monotonically parametrised error bounds include $b_{\mathbf{r}}(\mathbf{x}) = x_1^{r_1} + \dots + x_d^{r_d}$ and $b_{\mathbf{r}}(\mathbf{x}) = x_1^{r_1} \cdots x_d^{r_d}$, which are the sort of expressions that routinely appear in error bounds.
The proof of the following result can be found in \Cref{app: quasi likelihood sec}:

\begin{proposition}[Estimation using maximum quasi likelihood]
\label{prop: kernel parameter estimation}
Let $X_n^h = \{h \mathbf{x} : \mathbf{x} \in X_n\}$.
Suppose that $f \in \mathcal{H}_k(\mathcal{X})$ holds when $k$ in \eqref{eq: cov model} is based on the monotonically parametrised bound $b_{\mathbf{r}_0}(\mathbf{x})$ for some $\mathbf{r}_0 \geq \mathbf{0}$.
Let $\mathbf{K}_{b_{\mathbf{r}},h}$ denote the matrix with entries $k_{b_{\mathbf{r}}}(h \mathbf{x}_i , h \mathbf{x}_j)$, where the dependence of this matrix on both $h$ and $\mathbf{r}$ has now been emphasised, relative to the notation $\mathbf{K}_b$ introduced in \Cref{subsec: fitting the GP}.
Then any maximiser $\mathbf{r}_n^h[f] \in \argmax_{\mathbf{r} \geq \mathbf{0}} \mathcal{L}_n^h(\mathbf{r})$ of the log-quasi (marginal) likelihood
\begin{align}
  \mathcal{L}_n^h(\mathbf{r}) & := - f(X_n^h)^\top \mathbf{K}_{b_{\mathbf{r}},h}^{-1} f(X_n^h) +  \frac{ ( \mathbf{1}^\top \mathbf{K}_{b_{\mathbf{r}},h}^{-1} f(X_n^h) )^2 }{ \mathbf{1}^\top \mathbf{K}_{b_{\mathbf{r}},h}^{-1} \mathbf{1} } - \log \det \mathbf{K}_{b_{\mathbf{r}},h} \label{eq: QML}
\end{align}
satisfies $\liminf_{h \to 0} \mathbf{r}_n^h[f] \geq \mathbf{r}_0$.
\end{proposition}

\noindent The first two terms in \eqref{eq: QML} correspond to the (square of the) semi-norm $|m_n^h[f]|_{\mathcal{H}_k(\mathcal{X})}$, which is the analogue of the usual $\|m_n^h[h]\|_{\mathcal{H}_k(\mathcal{X})}$ term that we would appear in the likelihood had we not taken the $k_0^2 \rightarrow \infty$ limit; this justifies the interpretation of \eqref{eq: QML}, up to constants, as a quasi-likelihood. 
The one-sided conclusion of \Cref{prop: kernel parameter estimation} may be surprising at first, but this is in fact the strongest result that can be expected.
Indeed, the statement that $f(\mathbf{x}) - f(\mathbf{0}) = O(b_{\mathbf{r}_0}(\mathbf{x}))$ does not rule out the possibility that the error $f(\mathbf{x}) - f(\mathbf{0})$ decays \emph{faster} than $b_{\mathbf{r}_0}(\mathbf{x})$, and in this case we would expect the estimator $\mathbf{r}_n^h[f]$ to adapt to the actual convergence order.
The experiments that we report in \Cref{sec: applications} used maximum quasi (marginal) likelihood whenever convergence orders and/or kernel length-scale parameters were estimated.

\subsection{Generalisation to Multidimensional Output}
\label{subsec: multivar output}

Until this point we have considered the continuum quantity of interest $f(\mathbf{0})$ to be scalar-valued.
Oftentimes, however, we are interested in quantities $\{f(\mathbf{0},\mathbf{t})\}_{\mathbf{t} \in \mathcal{T}}$ that are vector- or function-valued depending on the nature of the index set $\mathcal{T}$.
The E-algorithm that we described in \Cref{sec: introduction} has been extended to finite-dimensional vector-valued output; see Chapter 4 of \citet{brezinski2013extrapolation} for detail.
A possible advantage of the \ac{GP}-based approach taken in \ac{GRE} is that it does not impose any mathematical structure on $\mathcal{T}$ beyond this being a set, making extension of the methodology to function-valued output straight-forward.

To extend our methodology to multivariate output, let $f : \mathcal{X} \times \mathcal{T} \rightarrow \mathbb{R}$ be such that $\{f(\mathbf{0},\mathbf{t})\}_{\mathbf{t} \in \mathcal{T}}$ is the continuum quantity of interest and $f(\mathbf{x},\mathbf{t})$ is a numerical approximation to $f(\mathbf{0},\mathbf{t})$.
For example, $f(0,t)$ may represent the solution to an ordinary differential equation at time $t$, while $f(x,t)$ may represent an approximation to this solution obtained using a Runge--Kutta method, with $x$ being the error tolerance of the Runge--Kutta method.
To improve presentation we will assume that $f(\mathbf{x},\mathbf{t}) - f(\mathbf{0},\mathbf{t}) = O(b(\mathbf{x}))$ \emph{uniformly} over $\mathbf{t} \in \mathcal{T}$, but $\mathbf{t}$-dependent error bounds could also be considered with additional notational overhead.
Our original covariance function \eqref{eq: cov model} can be generalised to
\begin{align}
    k((\mathbf{x},\mathbf{t}),(\mathbf{x}',\mathbf{t}')) = \sigma^2 \left[ k_0^2 + b(\mathbf{x}) b(\mathbf{x}') k_e(\mathbf{x},\mathbf{x}') \right] k_{\mathcal{T}}(\mathbf{t},\mathbf{t}') , \qquad \mathbf{x}, \mathbf{x}' \in \mathcal{X}, \; \mathbf{t},\mathbf{t}' \in \mathcal{T} , \label{eq: cov model 2}
\end{align}
where, to exploit tractable computation that results from this tensor product kernel, we have assumed a tensor product kernel and will assume that data $X_n = \{(\mathbf{x}_i,\mathbf{t}_j)\}_{i=1}^{n_1}{}_{j=1}^{n_2}$ are obtained on a Cartesian grid ($n = n_1 n_2$).
That is, with the data appropriately ordered we have the Kronecker decomposition $\mathbf{K} = \mathbf{K}_{\mathcal{X}} \otimes \mathbf{K}_{\mathcal{T}}$, where 
$\mathbf{K}_{\mathcal{X}}$ is the matrix with entries $k_0^2 + b(\mathbf{x}_i) b(\mathbf{x}_j) k_e(\mathbf{x}_i,\mathbf{x}_j)$, and $\mathbf{K}_{\mathcal{T}}$ is the matrix with entries $k_{\mathcal{T}}(\mathbf{t}_i,\mathbf{t}_j)$.
Then analogous calculations to those detailed in \Cref{app: flat prior limit}, which we present in \Cref{app: multivar calcs}, show that for values of $\mathbf{t},\mathbf{t}'$ contained in the dataset, the conditional mean and covariance function in the $k_0 \rightarrow \infty$ limit are
\begin{align*}
    m_n[f](\mathbf{x},\mathbf{t}) & = \left\{ \mathbf{k}_b(\mathbf{x})^\top \mathbf{K}_b^{-1} + [1 - \mathbf{k}_b(\mathbf{x})^\top \mathbf{K}_b^{-1} \mathbf{1}] \frac{\mathbf{1}^\top \mathbf{K}_b^{-1}}{\mathbf{1}^\top \mathbf{K}_b^{-1} \mathbf{1}} \right\} \otimes \left[ k_{\mathcal{T}}(\mathbf{t}) \mathbf{K}_{\mathcal{T}}^{-1} \right] f(X_n) \\
    k_n[f]((\mathbf{x},\mathbf{t}),(\mathbf{x}',\mathbf{t}')) & = \sigma_n^2[f] \Bigg\{ k_b(\mathbf{x},\mathbf{x}') k_{\mathcal{T}}(\mathbf{t},\mathbf{t}') - [\mathbf{k}_{\mathcal{T}}(\mathbf{t})^\top \mathbf{K}_{\mathcal{T}}^{-1} \mathbf{k}_{\mathcal{T}}(\mathbf{t}')] \times \\
    & \left. \qquad \qquad \left[ \mathbf{k}_b(\mathbf{x})^\top \mathbf{K}_b^{-1} \mathbf{k}_b(\mathbf{x}') - \frac{[\mathbf{k}_b(\mathbf{x})^\top \mathbf{K}_b^{-1} \mathbf{1} - 1] [\mathbf{k}_b(\mathbf{x}')^\top \mathbf{K}_b^{-1} \mathbf{1} - 1]^\top }{ \mathbf{1}^\top \mathbf{K}_b^{-1} \mathbf{1} } \right]  \right\} ,
\end{align*}
where $\mathbf{k}_{\mathcal{T}}(\mathbf{t})$ is the vector with entries $k_{\mathcal{T}}(\mathbf{t}_i,\mathbf{t})$.
For values of $\mathbf{t},\mathbf{t}'$ not contained in the training dataset, the conditional covariance does not have a finite limit; a proper prior should be used if off-grid prediction in the $\mathbf{t}$-domain is required.
Further details on the multivariate setting are deferred to \Cref{sec: applications}, where the approach is explored in the context of predicting temporal output from a cardiac model.

\subsection{Incorporating Additional Degrees of Freedom}
\label{subsec: extra parameters}

The final methodological extension that we consider is the case where $f_{\bm{\theta}}(\mathbf{x})$ additionally depends on one or more degrees of freedom $\bm{\theta} \in \Theta$; a setting where emulation or \ac{MFM} methods are routinely used (c.f. \Cref{sec: introduction}).
The proposed \ac{GRE} method can be applied in this context by viewing $f_{\bm{\theta}}(\mathbf{0})$ as a simulator with multidimensional output $\{ f(\mathbf{0},\bm{\theta}) \}_{\bm{\theta} \in \Theta}$ and then applying the methodology described in \Cref{subsec: multivar output} with $\bm{\theta}$, rather than $\mathbf{t}$, indexing the output of this extended model.
Since the required calculations are identical, we do not dwell any further on this point.

\bigskip

This completes our exposition of the \ac{GRE} method.
Next we next turn to a cardiac modelling case study, where the usefulness of the methodology is evaluated.

\section{Case Study: Cardiac Modelling}
\label{sec: applications}

The cardiac model $f_{\bm{\theta}}(\mathbf{x})$ that we consider in this section is a detailed numerical simulation of a single heart beat\footnote{The simulation is usually run until a steady-state is reached before reading off quantities of interest, at a substantial increase to the overall computational cost. 
For the present purpose we removed components from the model that required multiple heart beats to reach a steady state, and simulated only a single heart beat.} \citep{strocchi2023cell}. 
The simulation is rooted in finite element methods that require both a spatial ($x_1$) and a temporal ($x_2$) discretisation level to be specified; of these, the spatial discretisation is the most critical, due to the $O(x_1^{-3})$ cost associated with the construction of a suitable triangulation of the time-varying 3-dimensional volume of the heart; see \Cref{fig:BCs}.
The computational cost $c(\mathbf{x})$ is measured in real computational time (seconds) and comprises the \emph{setup time}, \emph{assembly time} (the time taken to assemble linear systems of equations), and the \emph{solver time} (the time taken to solve linear systems of equations), with assembly time the main contributor to total computational cost.
To achieve a clinically-acceptable level of accuracy, it is typical for a simulation $f_{\bm{\theta}}(\mathbf{x}_{\text{default}})$ to be performed with  $\mathbf{x}_{\text{default}} \approx (0.4 \text{ mm}, 2 \text{ ms})$, at a cost $c(\mathbf{x}_{\text{default}}) \approx 1.5 \times 10^4$ seconds (around $\approx 4$ hours) for a single heart beat\footnote{Simulations for this case study were performed on ARCHER2, a UK national super computing service (\url{https://www.archer2.ac.uk/}).
Each simulation involved 512 CPUs operating in parallel, so that simulation of one heart beat using setting $\mathbf{x}_{\text{default}}$ required $\approx 4 \times 512$ CPU hours in total.}.
This poses severe challenges to the scientific use of such models, with super-computing resources required to ascertain whether there are values of scientific parameters $\bm{\theta}$ for which observed data are consistent with model output \citep{strocchi2023cell}.
These challenges directly motivated the development of \ac{GRE}, and the remainder of this paper is dedicated to exploring the value of extrapolation methods in this context. 
Extrapolation of the cardiac model output represents a much greater challenge compared to extrapolation for the examples considered in \Cref{sec: methods}, due to the nonlinear physics being simulated.
Since our focus in this paper is not on inference for $\bm{\theta}$, these degrees of freedom were fixed to physically-realistic values based on previous analyses \citep{strocchi2020simulating,strocchi2023cell}, with all further details on the construction of the cardiac model reserved for \Cref{app: cardiac}.

\Cref{subsec: workflow} sets out a practical workflow for using the \ac{GRE} method, that focuses on the multidimensional setting where both convergence orders and kernel length-scale parameters are to be estimated.
The performance of \ac{GRE} is then investigated for both scalar-valued (\Cref{subsec: scalar cardiac}) and multivariate (\Cref{subsec: vector cardiac}) continuum quantities of interest.

\begin{figure}[t!]
    \centering
    \includegraphics[width=0.27\textwidth]{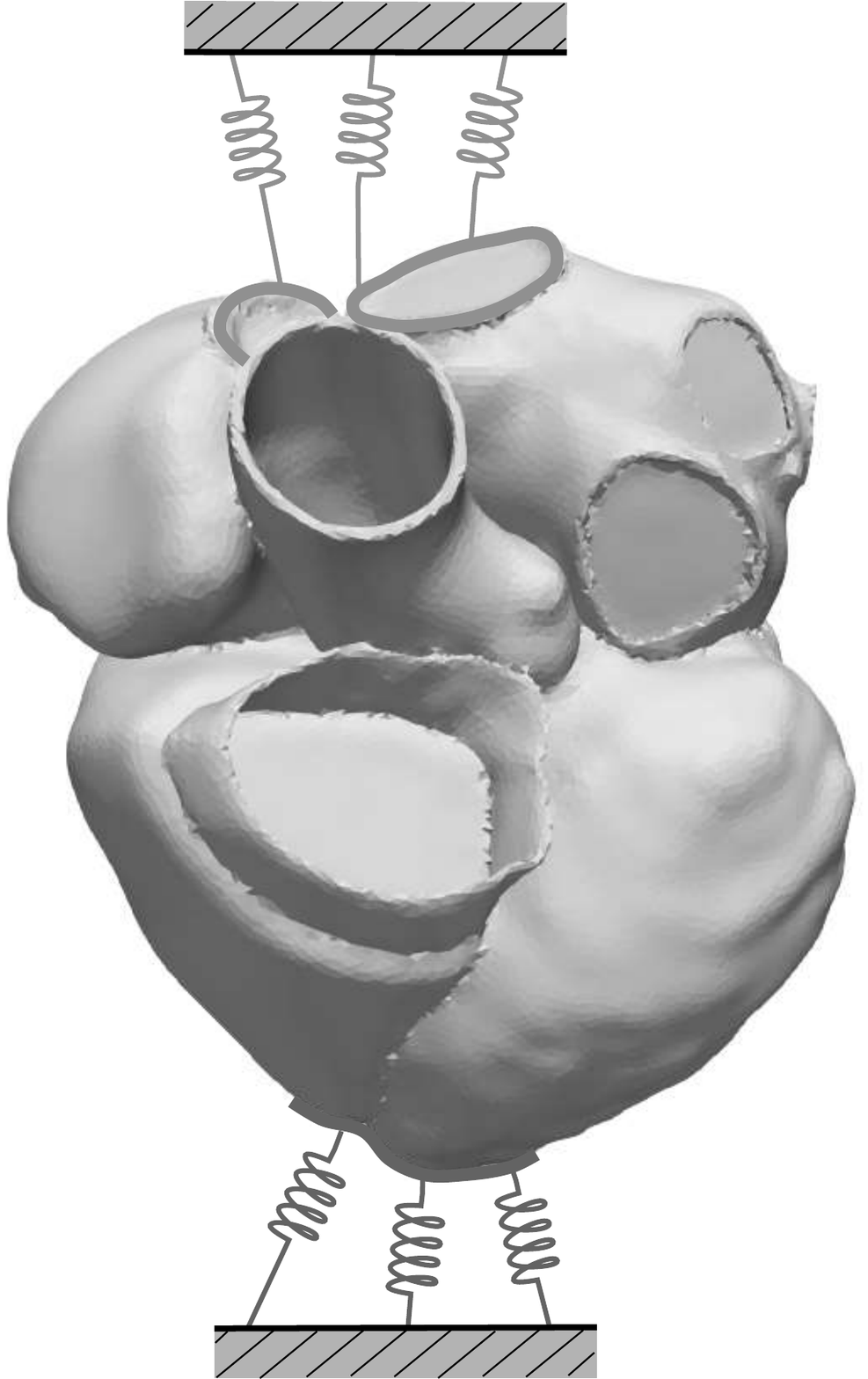}
    \includegraphics[width=0.7\textwidth,clip,trim = 0cm 0cm 0cm 1cm]{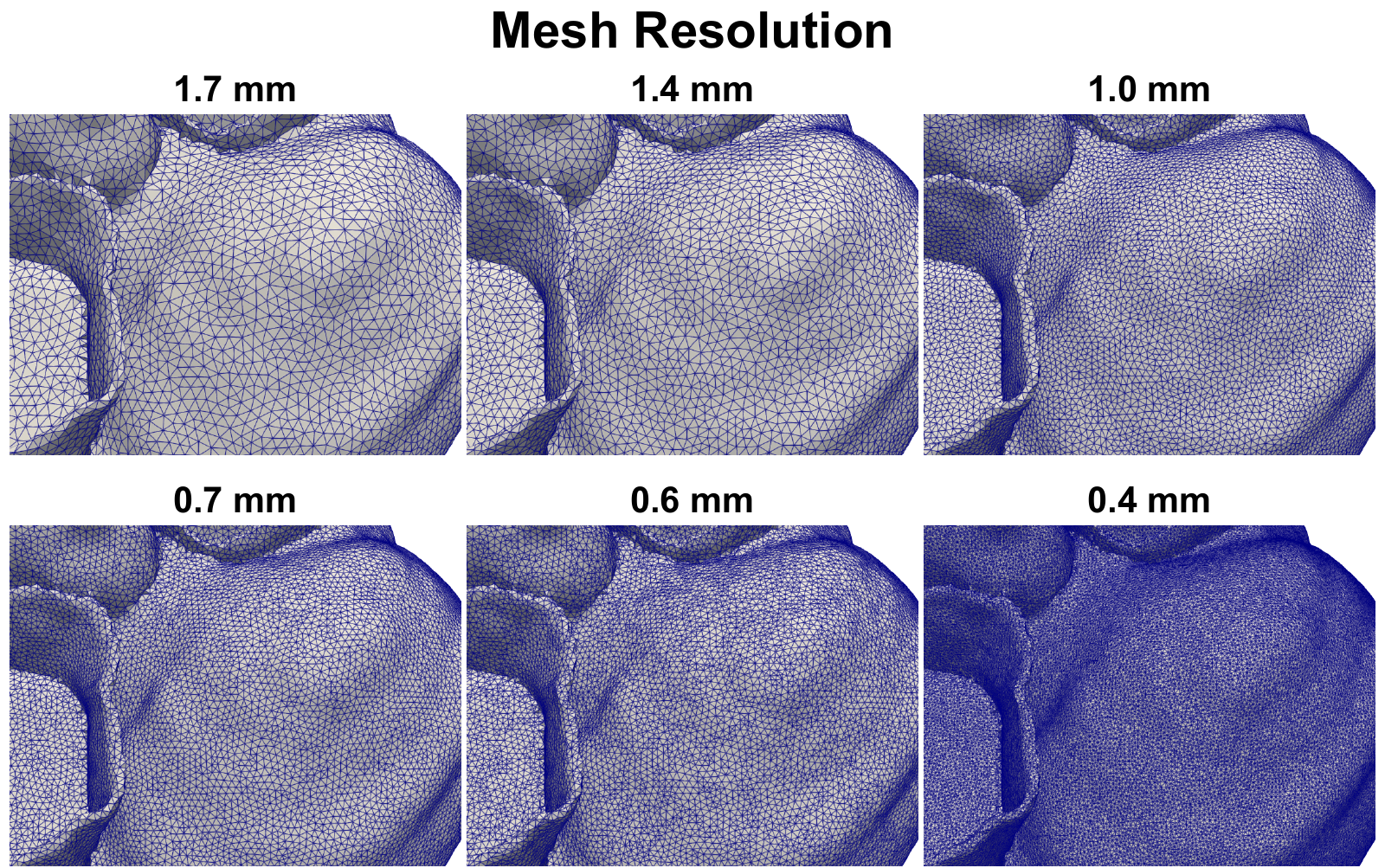}
    \caption{Cardiac model:  Left: Schematic indicating the veins and the apical region where spring boundary conditions were applied.
    Right: A subset of the mesh resolutions used in this case study.
    The finest resolution required $3 \times 10^7$ finite elements to be used.}
    \label{fig:BCs}
\end{figure}

\subsection{A Proposed General Workflow}
\label{subsec: workflow}

The sophistication of the cardiac model renders analytical derivation of convergence orders essentially impossible, so to proceed these orders must be estimated.
However, the computational cost of simulating from the model means that data from which convergence orders can be estimated are necessarily limited.
This motivates us to propose the following pragmatic workflow, which we present for a general model $f(\mathbf{x})$ and which scales in a reasonable way with the number $d$ of components of $\mathbf{x}$ that can be varied.
This workflow requires the user to specify a \ac{LOFI} setting $\mathbf{x}_{\text{lo-fi}}$ as a starting point, together with a means to predict the computational cost $c(\mathbf{x})$ of simulating $f(\mathbf{x})$, and a total computational budget $C$:

\begin{enumerate}
    \item For each fidelity parameter $x_i$, $i = 1,\dots,d$:
    \begin{enumerate}
        \item Simulate $f(\mathbf{x})$ for a range of values of $x_i$, with all of the other components $\mathbf{x}$ held fixed to their values in $\mathbf{x}_{\text{lo-fi}}$.
        \item Fit a univariate numerical analysis-informed \ac{GP} model (\ref{eq: gp post mean}, \ref{eq: gp post cov}) to these data, assuming an error bound of the form $b(x_i) = x_i^{r_i}$, where the scale estimate $\hat{\sigma}_i$ from \Cref{subsec: UQ} is used, and where the convergence order $r_i$, the kernel smoothness $s_i$, and the kernel length-scale parameter $\ell_i$ are simultaneously estimated using quasi maximum likelihood, as explained in \Cref{subsec: estimate order}.
    \end{enumerate}
    \item Construct a tensor product covariance model $k_e(\mathbf{x},\mathbf{x}') = k_e(x_1,x_1' ; \ell_1) \dots k_e(x_d,x_d' ; \ell_d)$
    and posit the overall error bound $b(\mathbf{x}) = \hat{\sigma}_1 x_1^{r_1} + \dots + \hat{\sigma}_d x_d^{r_d}$.  Then perform experimental design as described in \Cref{subsec: design}, with computational budget $C$.
    Denote the optimal design $X_n$.
    \item Simulate $f(\mathbf{x})$ for each $\mathbf{x} \in X_n$ and return the \ac{GRE} conditional mean \eqref{eq: simplified expressions} as the final approximation to $f(\mathbf{0})$.
\end{enumerate}

\noindent Several remarks are in order:
First, it is assumed that the Step 1 incurs negligible cost relative to the total computational budget; the precise interpretation of this assumption will necessarily be context-dependent.
Second, the additive form for $b(\mathbf{x})$ is appropriately conservative, in the sense that \emph{all} components of $\mathbf{x}$ must be small to control this bound.  
One could go further and compare the performance of \acp{GP} based on alternative form of $b(\mathbf{x})$, for example with interaction terms included, selecting among such models using maximum quasi likelihood, but for the present purposes the additive form of $b(\mathbf{x})$ is preferred since it is compatible with the independent estimation of convergence orders $r_i$ in Step 1.
Third, the independent estimation of $(r_i,s_i,\ell_i)$ for each $i = 1,\dots,d$ can be performed using brute-force search over a 3-dimensional grid to maximise the quasi-likelihood, whereas simultaneous estimation of all kernel parameters would be both statistically and computationally difficult.
The full workflow is demonstrated on our cardiac case study, next.

\begin{figure}[t!]
    \centering
    \includegraphics[width = 0.47\textwidth]{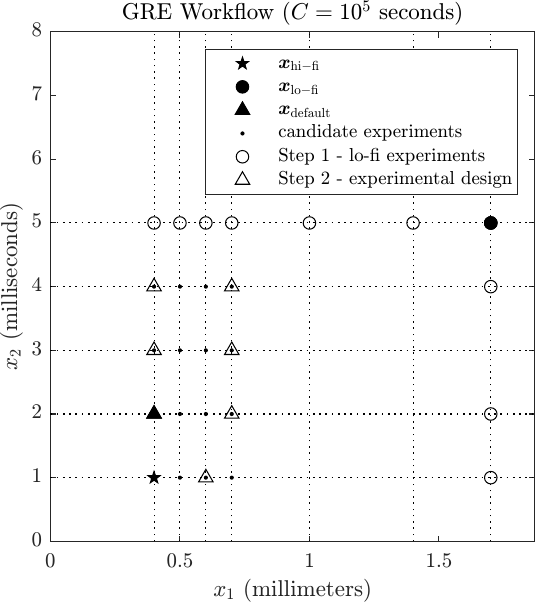}
    \includegraphics[width = 0.51\textwidth]{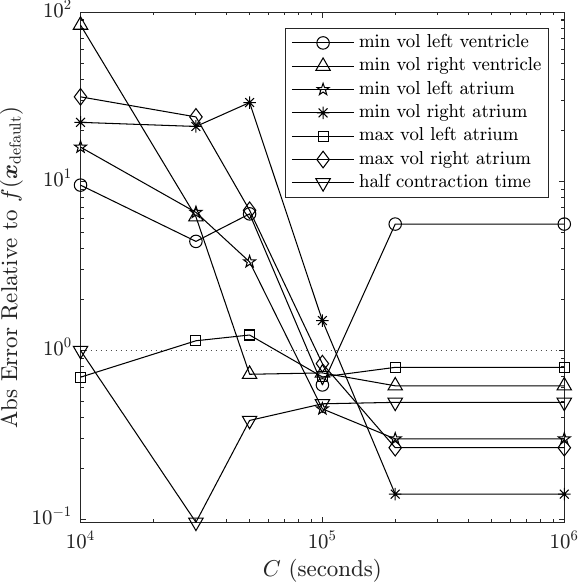}
    \caption{Scalar quantities of interest from the cardiac model.
    Left: The workflow, illustrated. In Step 1, the effect of varying each component of $\mathbf{x}$ in turn is explored, with all other components fixed equal to their value in $\mathbf{x}_{\text{lo-fi}}$.
    This facilitates the construction of a multivariate Gaussian process model for use in Step 2, where experimental design is performed (here shown for a computational budget of $C = 10^5$ seconds).
    For assessment purposes we aim to predict $f(\mathbf{x}_{\text{hi-fi}})$ as a ground truth, but in practice the goal is to predict $f(\mathbf{0})$.
    Right:  For each of the 7 scalar quantities of interest associated with the cardiac model we display the ratio of the absolute error $|f(\mathbf{x}_{\text{hi-fi}}) - m_n[f](\mathbf{x}_{\text{hi-fi}})|$ of the \ac{GRE} method and the absolute error $|f(\mathbf{x}_{\text{hi-fi}}) - f(\mathbf{x}_{\text{default}}) |$ of the default approximation, as a function of the total computational budget $C$.
    }
    \label{fig: workflow}
\end{figure}

\subsection{Approximation of Scalar Quantities of Interest}
\label{subsec: scalar cardiac}

The first part of our case study concerned the approximation of physiologically-interpretable scalar-valued quantities of interest.
These were; the minimum volume of the left and right ventricles and atria, the maximum volume during ventricular contraction for the left and right atria, and the time taken for the ventricles to contract in total capacity by one half; a total of 7 test problems for \ac{GRE}.

Though the computational time $c(\mathbf{x}_{\text{default}})$ is substantial, in this case study parallel computation resources can be exploited.
The main computational constraint that we work under here is that we will only run experiments for which $c(\mathbf{x}) \leq c(\mathbf{x}_{\text{default}})$ within our \ac{GRE} method.
Since the continuum limit $f(\mathbf{0})$ is intractable, we additionally computed a reference solution $f(\mathbf{x}_{\text{hi-fi}})$ with $\mathbf{x}_{\text{hi-fi}} = (0.4 \text{ mm}, 1 \text{ ms})$ and in what follows we assess how well the \ac{GRE} point estimate $m_n[f](\mathbf{x}_{\text{hi-fi}})$ approximates $f(\mathbf{x}_{\text{hi-fi}})$.
The central question here is whether the workflow proposed in \Cref{subsec: workflow} can provide more accurate approximation of $f(\mathbf{x}_{\text{hi-fi}})$ compared to $f(\mathbf{x}_{\text{default}})$, and if so what computational budget is required. 
To the best of our knowledge there do not exist comparable methodologies for this task; methods such as emulation and \ac{MFM} are not applicable when $\bm{\theta}$ is fixed, and classical extrapolation methods were not developed with multivariate $\mathbf{x}$ in mind.

The workflow is illustrated in the left panel of \Cref{fig: workflow}.
The \ac{LOFI} setting was $\mathbf{x}_{\text{lo-fi}} = (1.7 \text{ mm},5 \text{ ms})$. 
The convergence orders $r_1$, $r_2$ were selected from $\{0.5,1,2\}$, the smoothness parameters $s_1$, $s_2$ were selected from $\{0,1,2\}$, and the length-scales $\ell_1$, $\ell_2$ were selected using grid search, all estimated simultaneously using maximum quasi-likelihood.
Experimental designs were computed based on a candidate set of experiments, each of which incurs a cost no greater than $c(\mathbf{x}_{\text{default}})$, indicated by dots in the left panel of \Cref{fig: workflow}.
Rather than estimate the computational times, for this case study all candidate experiments were performed at the outset and their times recorded. 
Results for the 7 test problems are shown in the right panel of \Cref{fig: workflow}, where it is observed that the \ac{GRE} point estimator provides a generally better approximation to $f(\mathbf{x}_{\text{hi-fi}})$ compared to $f(\mathbf{x}_{\text{default}})$ when the computational budget $C$ reaches or exceeds $10^5$.
The optimal design for approximating the minimum volume of the left ventricle is depicted in the left hand panel of \Cref{fig: workflow} for a computational budget $C = 10^5$; the design supplements $\mathbf{x}_{\text{default}}$ with 6 additional simulations of lower cost, analogous to a classical extrapolation method but here generalised to the multivariate context.
Note that for $C$ exceeding $2 \times 10^5$ the optimal design becomes saturated, containing all experiments in the candidate set.
That \ac{GRE} should perform worse than $f(\mathbf{x}_{\text{defult}})$ at small computational budgets is not surprising given that all convergence orders $r_i$, smoothnesses $s_i$, and length-scales $\ell_i$ are estimated from the small \ac{LOFI} training dataset, and these values largely determine the output of \ac{GRE} in the absence of a sufficient number of experiments in the training dataset $X_n$.
However, for a sufficiently large computational budget, it is encouraging to see that information from the experiments in $X_n$, each of which cost no greater than $c(\mathbf{x}_{\text{default}})$, is exploited in \ac{GRE} to achieve more accurate estimation for 6 of the 7 scalar quantities of interest.

\subsection{Approximation of Temporal Model Output}
\label{subsec: vector cardiac}

\begin{figure}[t!]
    \centering
    \includegraphics[width = \textwidth]{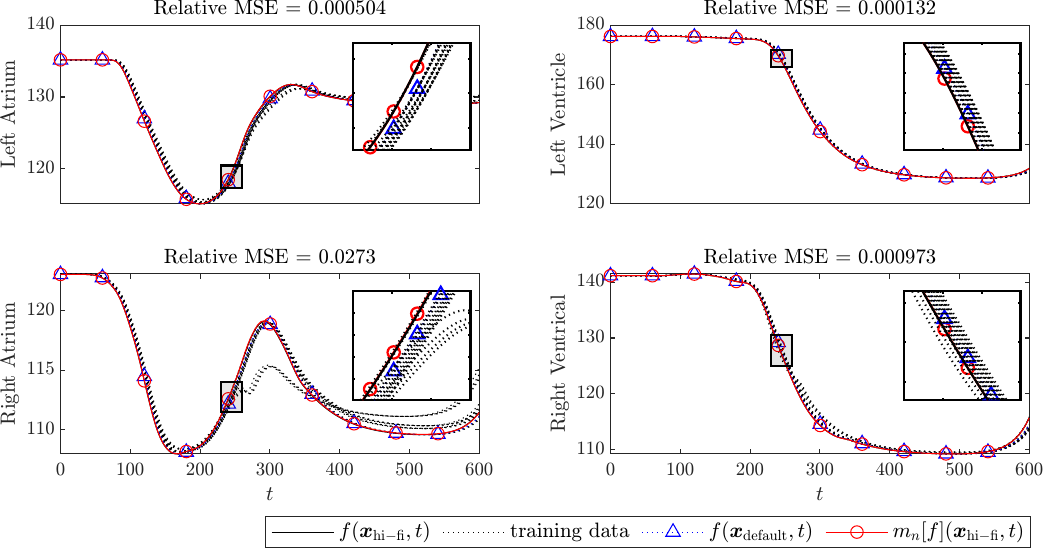}
    \caption{Temporal quantities of interest from the cardiac model.
             For each of the 4 temporal quantities of interest associated with the cardiac model we display the approximations produced at the different spatial and temporal resolutions in the training dataset $X_n$, together with the ground truth $f(\mathbf{x}_{\text{hi-fi}})$ (solid black lines), the default approximation $f(\mathbf{x}_{\text{default}},t)$ (blue triangles), and the approximation $m_n[f](\mathbf{x}_{\text{hi-fi}},t)$ from Gauss--Richardson Extrapolation (GRE; red circles).
             The ratio of the mean square error $\int [ f(\mathbf{x}_{\text{hi-fi}},t) - m_n[f](\mathbf{x}_{\text{hi-fi}},t) ]^2 \; \mathrm{d}t$ of the GRE method and the mean square error $\int [ f(\mathbf{x}_{\text{hi-fi}},t) - f(\mathbf{x}_{\text{default}},t) ]^2 \; \mathrm{d}t$ of the default method is also reported.}
    \label{fig: temporal results}
\end{figure}

The scalar quantities of interest considered in \Cref{subsec: scalar cardiac} are summary statistics obtained from 4-dimensional temporal model output of the form $\mathbf{f}(\mathbf{x},t)$, where here $t$ is a time index ranging from 0 to 600 ms and the components of $\mathbf{f}$ refer to the volumes of the atria and ventricles.
It is therefore interesting to investigate whether these temporal outputs can be directly approximated, providing 4 test problems for the methodology described in \Cref{subsec: multivar output}.
Here for simplicity we fixed the discrete values $s_1$, $s_2$, $r_1$, and $r_2$, the median of the values estimated in \Cref{subsec: scalar cardiac}, and we fixed the continuous values $\hat{\sigma}_1$, $\hat{\sigma}_2$, $\ell_1$, and $\ell_2$ to the mean of the values estimated in \Cref{subsec: scalar cardiac}.
The length-scale for the kernel $k_{\mathcal{T}}$ was set equal to the length of the time series itself.
The computational budget was fixed to $C = 2 \times 10^5$, so that our experimental design is saturated, but recall that no individual experiment in this design had cost exceeding $c(\mathbf{x}_{\text{default}})$.
Full results are displayed in \Cref{fig: temporal results}.
In each case the approximation produced by \ac{GRE} achieves lower mean square error relative to $f(\mathbf{x}_{\text{default}},t)$.
Taken together with the results in \Cref{subsec: scalar cardiac}, these results are an encouraging and pave the way for subsequent investigations and applications of \ac{GRE}.

\section{Discussion}
\label{sec: discuss}

This paper introduced a probabilistic perspective on extrapolation, presenting a framework in which classical extrapolation methods from numerical analysis and modern multi-fidelity modelling are unified.
One approach was developed in detail, which we termed \acf{GRE}.
The \ac{GRE} method facilitates simultaneous convergence acceleration and uncertainty quantification, and unlocks experimental design functionality for optimisation over the set of fidelities at which simulation is performed.
The end result is a methodology that allows a practitioner to arrive, in a principled manner, at fidelities $\{\mathbf{x}_i\}_{i=1}^n$ such that the associated simulator outputs $\{f(\mathbf{x}_i)\}_{i=1}^n$ can be combined to produce an approximation to the continuum quantity $f(\mathbf{0})$ that is typically more accurate than a single \ac{HIFI} simulation run at a comparable computational cost.
A cardiac modelling case study provided an initial positive proof-of-concept, but further case studies -- involving different types of computer model -- will be required to comprehensively assess \ac{GRE}; we aim to undertake domain-specific investigations in future work.

Several methodological extensions to this work can be envisaged, such as considering alternative regression models to \acp{GP}, developing theory and methodology for the more challenging cases where the regression model is misspecified and computational costs needs to be predicted, and extending the experimental design methodology to include additional degrees of freedom $\bm{\theta}$, which are often present in a mathematical model.
In addition, and more speculatively, it would be interesting to explore modern computational tasks, such as the super-resolution task in deep learning, for which extrapolation methods have yet to be exploited.

\paragraph{Acknowledgements}

The authors wish to thank Onur Teymur and Jere Koskela for early discussions of this project.
CJO was supported by EP/W019590/1 and a Leverhulme Prize.
TK was supported by the Research Council of Finland postdoctoral researcher grant number 338567, \textit{Scalable, Adaptive and Reliable Probabilistic Integration}. ALT was supported by EP/X01259X/1 and EP/R014604/1. CJO and ALT would like to thank the Isaac Newton Institute for Mathematical Sciences, Cambridge, for support and hospitality during the programme {\em Mathematical and Statistical Foundation of Future Data-Driven Engineering} where work on this paper was undertaken. 
MS and SAN were supported by the Wellcome/EPSRC Centre for Medical Engineering (WT203148/Z/16/Z). 
SAN was supported by NIH R01-HL152256, ERC PREDICT-HF 453 (864055), BHF (RG/20/4/34803), and EPSRC (EP/P01268X/1, EP/X03870X/1).

\newpage
\appendix

\part{} 
\parttoc 

\section{Kernels and Smoothness Spaces}
\label{app: kernels}

This appendix contains definitions for the kernels $k_e : \mathcal{X} \times \mathcal{X} \rightarrow \mathbb{R}$, on a bounded set $\mathcal{X} \subset \mathbb{R}^d$, referred to in the main text, together with details about the smoothness of elements in the Hilbert spaces $\mathcal{H}_{k_e}(\mathcal{X})$ that are reproduced.
All kernels $k_e$ that we discussed take the radial form 
$$
k_e(\mathbf{x},\mathbf{y}) = \phi\left( d_{\bm{\ell}}(\mathbf{x},\mathbf{y}) \right), \qquad d_{\bm{\ell}}(\mathbf{x},\mathbf{y}) =  \left( \sum_{i=1}^d \frac{(x_i - y_i)^2}{\ell_i^2} \right)^{1/2} , 
$$
where the radial function $\phi : [0,\infty) \rightarrow \mathbb{R}$ and the length-scale parameters $\bm{\ell} \in (0,\infty)^d$ are to be specified.

\paragraph{M\'{a}tern Kernels}

The M\'{a}tern family of kernels is defined via the radial function
$$
\phi_\nu(z) = \exp\left( - \sqrt{2s+1} \; z \right) \frac{s!}{(2s)!} \sum_{i=0}^s \frac{(s+i)!}{i! (s-i)!} \left( 2 \sqrt{2s+1} \; z \right)^{s-i} 
$$
for which $k_e \in C^{2s}(\mathcal{X} \times \mathcal{X})$ \citep[][Section 2.7]{stein1999interpolation}. 
The kernel $k_e$ reproduces (up to an equivalent norm) the Sobolev space \smash{$H^{s + \frac{d+1}{2}}(\mathcal{X})$} \citep[][Corollary 10.48]{wendland2004scattered}.
The elements of \smash{$H^{s + \frac{d+1}{2}}(\mathcal{X})$} are functions whose mixed weak partial derivatives up to order $s + \frac{d+1}{2}$ exist as elements of $L^2(\mathcal{X})$, and so in particular \smash{$C^{s + \frac{d+1}{2}}(\mathcal{X}) \subset H^{s + \frac{d+1}{2}}(\mathcal{X})$}.
Conversely, from the Sobolev embedding theorem, \smash{$H^{s + \frac{d+1}{2}}(\mathcal{X}) \subset C^s(\mathcal{X})$}.
Thus \smash{$C^{s + \frac{d+1}{2}}(\mathcal{X}) \subset \mathcal{H}_{k_e}(\mathcal{X}) \subset C^s(\mathcal{X})$}.
It may be helpful to point out that, somewhat confusingly, in most of the statistical literature the M\'{a}tern kernel is defined in terms of a `smoothness parameter' $\nu \coloneqq s + \frac{1}{2}$ \citep{porcu2023mat}.
From \Cref{subsec: faster convergence} of this paper onward, ``Mat\'{e}rn ($s = m$)'' refers to $\phi_\nu$ with $\nu = m + \frac{1}{2}$.
As an example, the Mat\'{e}rn-$\frac{5}{2}$ kernel that appears in \Cref{fig: illustrating samples} would correspond to $s = 2$ in our framework.

\paragraph{Wendland Kernels}

Let $z_+^m \coloneqq \max(0,z)^m$ and $(\mathcal{I} \varphi)(z) \coloneqq \int_z^\infty t \varphi(t) \; \mathrm{d}t$.
The Wendland family of kernels is defined via the radial function
$$
\phi_{d,s}(z) = \mathcal{I}^s \varphi_{\lfloor \frac{d}{2} \rfloor + s + 1} (z), \qquad \varphi_m(z) = (1-z)_+^m
$$
for $s \in \mathbb{N}_0$, for which $k_e \in C^{2s}(\mathcal{X} \times \mathcal{X})$; see Theorem 9.13 of \citet{wendland2004scattered}.
The kernel $k_e$ reproduces (up to an equivalent norm) the Sobolev space \smash{$H^{s + \frac{d+1}{2}}(\mathcal{X})$}, at least when $s \geq 1$ (for $s=0$ we need $d \geq 3$); see Theorem 10.35 of \citet{wendland2004scattered}. 
Thus again \smash{$C^{s + \frac{d+1}{2}}(\mathcal{X}) \subset \mathcal{H}_{k_e}(\mathcal{X}) \subset C^s(\mathcal{X})$}.
Wendland kernels are sometimes preferred to Mat\'{e}rn kernels due to their compact support \citep[but they are not the only available alternative with compact support;][]{porcu2023mat}.

\paragraph{Gaussian Kernels}

More generally, if the kernel $k_e$ satisfies $k_e \in C^{2s}(\mathcal{X} \times \mathcal{X})$ with $\mathcal{X}$ an open subset of $\mathbb{R}^d$, then $\mathcal{H}_{k_e}(\mathcal{X}) \subset C^s(\mathcal{X})$; see Theorem 10.45 of \citet{wendland2004scattered}.
In particular the Gaussian kernel, defined by the radial function
$$
\phi(z) = \exp\left(-z^2\right) , 
$$
has continuous derivatives of all orders, so elements of the associated Hilbert space are $C^\infty(\mathcal{X})$.

\section{Proofs of Results in the Main Text}

This appendix contains proofs for all theoretical results presented in the main text.

\subsection{Sample Path Properties of Numerical Analysis-Informed GPs}
\label{subsec: sample path}

This appendix presents sufficient conditions under which sample paths $g$ from the numerical analysis-informed \ac{GP} in \eqref{eq: cov model} satisfy, with probability one, $g(\mathbf{x}) - g(\mathbf{0}) = O(b(\mathbf{x}))$ in the $\mathbf{x} \rightarrow \mathbf{0}$ limit. 
To be precise, we establish that there is a  \emph{version} $g$ of $\mathcal{GP}(0,k)$ for which, with probability one, $g(\mathbf{x}) - g(\mathbf{0}) = O(b(\mathbf{x}))$.  
Recall that a stochastic process $h$ is said to be a \emph{version} of $g$ if $h(\mathbf{x}) = g(\mathbf{x})$ with probability one, for all $\mathbf{x} \in \mathcal{X} \setminus \{\mathbf{0}\}$.
This technical consideration arises simply because stochastic processes can be altered on null sets, affecting their convergence properties while leaving their distribution unchanged.

\begin{proposition}[Sample paths of numerical analysis-informed \acp{GP}]
In the setting of \Cref{subsec: prior}, assume that $k_e$ is H\"{o}lder continuous for some exponent; i.e. there exist $C, \alpha > 0$ such that $|k_e(\mathbf{x},\mathbf{x}) - k_e(\mathbf{x},\mathbf{x}')| \leq C \|\mathbf{x} - \mathbf{x}'\|^{\alpha} 
$ for all $\mathbf{x},\mathbf{x}' \in \mathcal{X}$, and that the set $\mathcal{X}$ is bounded.
Then there is a version $g$ of $\mathcal{GP}(0,k)$ for which, with probability one, $g(\mathbf{x}) - g(\mathbf{0}) = O(b(\mathbf{x}))$ in the $\mathbf{x} \rightarrow \mathbf{0}$ limit.
\end{proposition}
\begin{proof}
Let $e(\mathbf{x}) := b(\mathbf{x})^{-1} (g(\mathbf{x}) - g(\mathbf{0}))$ for $\mathbf{x} \in \mathcal{X} \setminus \{\mathbf{0}\}$, which is well-defined since $b(\mathbf{x}) > 0$ for all $\mathbf{x} \in \mathcal{X} \setminus \{\mathbf{0}\}$. 
From \eqref{eq: cov model}, if $g \sim \mathcal{GP}(0,k)$ then the distribution of $e$ is $\mathcal{GP}(0,k_e|_{\mathcal{X} \setminus \{\mathbf{0}\}})$, where $k_e|_{\mathcal{X} \setminus \{\mathbf{0}\}}(\mathbf{x},\mathbf{x}') = k_e(\mathbf{x},\mathbf{x}')$ for all $\mathbf{x},\mathbf{x}' \in \mathcal{X} \setminus \{\mathbf{0}\}$.
Our task is equivalent to establishing that there is a version of $e$ for which $e(\mathbf{x})$ is almost surely bounded as $\mathbf{x} \rightarrow \mathbf{0}$, but we will in fact establish the stronger result that there is a version of $e$ for which $\mathbb{P}( \sup_{\mathbf{x} \in \mathcal{X} \setminus \{\mathbf{0}\}} |e(\mathbf{x})| < \infty ) = 1$.
From Markov's inequality, this follows if $\mathbb{E}[\sup_{\mathbf{x} \in \mathcal{X} \setminus \{\mathbf{0}\}} |e(\mathbf{x})|] < \infty$, and in fact it is sufficient to show that
\begin{align}
\mathbb{E}\left[ \sup_{\mathbf{x} \in \mathcal{X} \setminus \{\mathbf{0}\}} e(\mathbf{x}) \right] < \infty . \label{eq: sufficient path}
\end{align}
Indeed, for any $\mathbf{x}_0 \in \mathcal{X} \setminus \{\mathbf{0}\}$, we have 
\begin{align*}
    |e(\mathbf{x})| & \leq |e(\mathbf{x}_0)| + |e(\mathbf{x}) - e(\mathbf{x}_0)| \\
    & \leq |e(\mathbf{x}_0)| + \sup_{\mathbf{x}' \in \mathcal{X} \setminus \{\mathbf{0}\}} [ e(\mathbf{x}') - e(\mathbf{x}_0) ] - \inf_{\mathbf{x}' \in \mathcal{X} \setminus \{\mathbf{0}\}} [ e(\mathbf{x}') - e(\mathbf{x}_0) ] 
\end{align*}
from which it follows that
\begin{align*}
    \mathbb{E}\left[ \sup_{\mathbf{x} \in \mathcal{X} \setminus \{\mathbf{0}\}} |e(\mathbf{x})| \right] & \leq \mathbb{E}[|e(\mathbf{x}_0)|] + \mathbb{E}\left[\sup_{\mathbf{x}' \in \mathcal{X} \setminus \{\mathbf{0}\}} e(\mathbf{x}') - e(\mathbf{x}_0) \right] - \mathbb{E}\left[\inf_{\mathbf{x}' \in \mathcal{X} \setminus \{\mathbf{0}\}} e(\mathbf{x}') - e(\mathbf{x}_0) \right] \\
    & = \mathbb{E}[|e(\mathbf{x}_0)|] + 2 \mathbb{E}\left[\sup_{\mathbf{x}' \in \mathcal{X} \setminus \{\mathbf{0}\}} e(\mathbf{x}') \right] ,
\end{align*}
where $\mathbb{E}[|e(\mathbf{x}_0)|] < \infty$ since $e(\mathbf{x}_0)$ is Gaussian, and where the final equality followed from $\mathbb{E}[e(\mathbf{x}_0)] = 0$ and symmetry of the \ac{GP}.

Our main tools to establish \eqref{eq: sufficient path} are entropy numbers and Dudley's theorem.
The H\"{o}lder condition ensures that the induced pseudometric $\mathsf{d} : \mathcal{X} \times \mathcal{X} \rightarrow [0,\infty)$ defined via
$\mathsf{d}(\mathbf{x},\mathbf{x}')^2 \coloneqq k_e(\mathbf{x},\mathbf{x}) - 2 k_e(\mathbf{x},\mathbf{x}') + k_e(\mathbf{x}',\mathbf{x}')$
satisfies $\mathsf{d}(\mathbf{x},\mathbf{x}') \leq \sqrt{2C} \|\mathbf{x} - \mathbf{x}'\|^{\alpha/2}$ for all $\mathbf{x},\mathbf{x}' \in \mathcal{X}$.
Let $B_{\epsilon,\mathsf{d}}(\mathbf{x}) \coloneqq \{\mathbf{x}' \in \mathcal{X} \setminus \{\mathbf{0}\} : \mathsf{d}(\mathbf{x},\mathbf{x}') < \epsilon\}$ denote an open $\mathsf{d}$-ball of radius $\epsilon$ centred at $\mathbf{x}$, and let $N(\mathcal{X} \setminus \{\mathbf{0}\},\mathsf{d};\epsilon)$ denote the \emph{entropy number}; the minimal number of open $\mathsf{d}$-balls of radius $\epsilon$ required to cover $\mathcal{X} \setminus \{\mathbf{0}\}$.
The boundedness condition on $\mathcal{X}$ ensures that the entropy number is well-defined.
Dudley's theorem states that, in our context, there is a version $e$ of $\mathcal{GP}(0,k_e)$ such that
\begin{align}
\mathbb{E}\left[ \sup_{\mathbf{x} \in \mathcal{X}  \setminus \mathbf{0}} e(\mathbf{x}) \right] \leq 24 \int_0^\infty \sqrt{\log N(\mathcal{X} \setminus \{\mathbf{0}\},\mathsf{d};\epsilon)} \; \mathrm{d}\epsilon ; \label{eq: dudley}
\end{align}
see Theorem 11.17 in \citet{ledoux1991probability}.
Our task is now to establish that the integral in \eqref{eq: dudley} is finite, and to this end we make use of a simple upper bound on the entropy number for $\mathsf{d}$ in terms of the entropy number for the usual Euclidean distance $\mathsf{e}(\mathbf{x},\mathbf{x}') \coloneqq \|\mathbf{x} - \mathbf{x}'\|$ by noting that
\begin{align*}
    B_{\epsilon,\mathsf{d}}(\mathbf{x}) & = \{\mathbf{x}' \in \mathcal{X} \setminus \{\mathbf{0}\} : \mathsf{d}(\mathbf{x},\mathbf{x}') < \epsilon\} \\
    & \supseteq \{\mathbf{x}' \in \mathcal{X} \setminus \{\mathbf{0}\} : \sqrt{2C} \|\mathbf{x} - \mathbf{x}'\|^{\alpha/2} < \epsilon\}
    = B_{(\epsilon/C)^{1/\alpha}, \mathsf{e}}(\mathbf{x}) ,
\end{align*}
from which it follows that $N(\mathcal{X} \setminus \{\mathbf{0}\},\mathsf{d};\epsilon) \leq N(\mathcal{X} \setminus \{\mathbf{0}\} ,\mathsf{e};(\epsilon/C)^{1/\alpha})$.
Since $\mathcal{X}$ is bounded, the entropy number for the Euclidean distance can be bounded by considering the number of Euclidean balls of radius $(\epsilon/C)^{1/\alpha}$ that are needed to cover a sufficiently large cube in $\mathbb{R}^d$, from which we obtain a bound of the form
$$
N\left(\mathcal{X} \setminus \{\mathbf{0}\},\mathsf{e}; \left(\frac{\epsilon}{C} \right)^{1/\alpha} \right) \leq  \max\left\{ 1 , \tilde{C} \left( \frac{C}{\epsilon} \right)^{d/\alpha} \right\} 
$$
for some constant $\tilde{C}$, from which the finiteness of the integral in \eqref{eq: dudley} can be established.
Indeed, letting $\epsilon_0 = \tilde{C}^{\alpha/d} C$, we have the bound
\begin{align*}
    \int_0^\infty \sqrt{\log N(\mathcal{X},\mathsf{d};\epsilon)} \; \mathrm{d}\epsilon & \leq \int_0^\infty \sqrt{ \max\left\{ 0 , \frac{d}{\alpha} \log\left( \frac{\epsilon_0}{\epsilon} \right)  \right\} } \; \mathrm{d}\epsilon \\
    & = \sqrt{\frac{d}{\alpha}} \int_0^{\epsilon_0} \sqrt{\log\left( \frac{\epsilon_0}{\epsilon} \right)} \; \mathrm{d}\epsilon = \sqrt{\frac{\pi d}{\alpha}} \frac{\epsilon_0}{2}
    < \infty ,
\end{align*}
as required.
\end{proof}

\noindent The H\"{o}lder and boundedness assumptions in \Cref{prop: discrete} are weak and hold for all of the examples in this paper that we considered.

\subsection{Derivation of the `Objective' Prior Limit}
\label{app: flat prior limit}

This appendix contains standard calculations that can be found in references such as  \citet{karvonen2018bayes}, but we include them here to keep the paper self-contained. 
Let $f \sim \mathcal{GP}(0,k)$ where we fix finite values of $\sigma^2, k_0^2 > 0$ in the specification of $k$ in \eqref{eq: cov model}.
The distribution of $f$ conditional upon the components $f(X_n)$ takes the familiar form
\begin{align}
m_n[f](\mathbf{x}) & = \mathbf{k}(\mathbf{x})^\top \mathbf{K}^{-1} f(X_n),  \label{eq: standard GP mean} \\
k_n[f](\mathbf{x},\mathbf{x}') & = k(\mathbf{x},\mathbf{x}') - \mathbf{k}(\mathbf{x})^\top \mathbf{K}^{-1} \mathbf{k}(\mathbf{x}') ; \label{eq: standard GP cov}
\end{align}
see Chapter 2 of \citet{rasmussen2006gaussian}.
Here $k(\mathbf{x},\mathbf{x}') = \sigma^2 \{k_0^2 + k_b(\mathbf{x},\mathbf{x}') \}$, $\mathbf{k}(\mathbf{x}) = \sigma^2 \{ k_0^2 \mathbf{1} + \mathbf{k}_b(\mathbf{x}) \}$, and $\mathbf{K} = \sigma^2 \{ k_0^2 \mathbf{1} \mathbf{1}^\top + \mathbf{K}_b \}$.
Next we use the Woodbury matrix inversion identity to deduce that
\begin{align}
\mathbf{K}^{-1} = \sigma^{-2} (k_0^2 \mathbf{1} \mathbf{1}^\top + \mathbf{K}_b)^{-1} & = \sigma^{-2} \{ \mathbf{K}_b^{-1} - \mathbf{K}_b^{-1} \mathbf{1} ( k_0^{-2} + \mathbf{1}^\top \mathbf{K}_b^{-1} \mathbf{1} )^{-1} \mathbf{1}^\top \mathbf{K}_b^{-1} \} . \label{eq: after Woodbury}
\end{align}
Plugging this into \eqref{eq: standard GP mean} and \eqref{eq: standard GP cov}, we obtain 
\begin{align}
    m_n[f](\mathbf{x}) & = \{k_0^2 \mathbf{1} + \mathbf{k}_b(\mathbf{x})\}^\top \{ \mathbf{K}_b^{-1} - \mathbf{K}_b^{-1} \mathbf{1} ( k_0^{-2} + \mathbf{1}^\top \mathbf{K}_b^{-1} \mathbf{1} )^{-1} \mathbf{1}^\top \mathbf{K}_b^{-1} \} f(X_n), \nonumber \\ 
    \frac{k_n[f](\mathbf{x},\mathbf{x}')}{\sigma^2} & = k_0^2 + k_b(\mathbf{x},\mathbf{x}') - (k_0^2 \mathbf{1} + \mathbf{k}_b(\mathbf{x}))^\top \left\{ 
\mathbf{K}_b^{-1} - \frac{ \mathbf{K}_b^{-1} \mathbf{1}  \mathbf{1}^\top \mathbf{K}_b^{-1} }{ ( k_0^{-2} + \mathbf{1}^\top \mathbf{K}_b^{-1} \mathbf{1} ) } \right\} (k_0^2 \mathbf{1} + \mathbf{k}_b(\mathbf{x}')) . \label{eq: plug in k}
\end{align}
Then, for small $k_0^{-2}$, we have from a Taylor expansion that
\begin{align}
    (k_0^{-2} + \mathbf{1}^\top \mathbf{K}_b^{-1} \mathbf{1})^{-1} = \frac{1}{ \mathbf{1}^\top \mathbf{K}_b^{-1} \mathbf{1} }  - \frac{k_0^{-2}}{ (\mathbf{1}^\top \mathbf{K}_b^{-1} \mathbf{1})^2 } + \frac{k_0^{-4}}{ (\mathbf{1}^\top \mathbf{K}_b^{-1} \mathbf{1})^3 } + O(k_0^{-6})  \label{eq: Taylor for denominator}
\end{align}
so that
\begin{align*}
m_n[f](\mathbf{x}) ={}&  \frac{\mathbf{1}^\top \mathbf{K}_b^{-1} f(X_n) }{ \mathbf{1}^\top \mathbf{K}_b^{-1} \mathbf{1} } + \mathbf{k}_b(\mathbf{x})^\top \mathbf{K}_b^{-1} \left\{ f(X_n) -  \left( \frac{\mathbf{1}^\top \mathbf{K}_b^{-1} f(X_n) }{ 
  \mathbf{1}^\top \mathbf{K}_b^{-1} \mathbf{1} } \right) \mathbf{1} \right\} + O(k_0^{-2}),  \\
k_n[f](\mathbf{x},\mathbf{x}') ={}& \sigma^2 \left\{ k_b(\mathbf{x},\mathbf{x}') - \mathbf{k}_b(\mathbf{x})^\top \mathbf{K}_b^{-1} \mathbf{k}_b(\mathbf{x}') + \frac{ [ \mathbf{k}_b(\mathbf{x})^\top \mathbf{K}_b^{-1} \mathbf{1} - 1 ] [ \mathbf{k}_b(\mathbf{x}')^\top \mathbf{K}_b^{-1} \mathbf{1} - 1 ]^\top }{ \mathbf{1}^\top \mathbf{K}_b^{-1} \mathbf{1} } \right\} \\
&+ O(k_0^{-2}) , 
\end{align*}
which gives the stated limiting result.

\subsection{Technical Results on Polynomial Reproduction}

The proof that we present for \Cref{prop: mean conver} is based on local polynomial reproduction, and we aim for a result similar to Theorem 11.13 in \citet{wendland2004scattered}.
That result, however, aims for applicability to general domains and relies on an \emph{interior cone condition} to ensure that fill distances based on balls can be meaningfully related to the approximation task.
In our case, where the domain is axis-aligned, simpler and sharper results may be obtained by performing an alternative analysis based directly on boxes instead.
Furthermore, we require a result for non-isotropic kernels, while Theorem 11.13 in \citet{wendland2004scattered} assumes an isotropic kernel.

Our starting points is the following fundamental result on polynomial functions and scattered data:

\begin{lemma}[Lemma 1 in \citet{madych1992bounds}] \label{lem: technical}
Let $\gamma_1 = 2$ and $\gamma_d = 2d(1+\gamma_{d-1})$ for $d > 1$.
Let $\ell \in \mathbb{N}$ with $q \geq \gamma_d(\ell+1)$.
Let $\mathcal{X} = [\mathbf{a}, \mathbf{a} + \lambda \mathbf{1}]$ for some $\mathbf{a} \in \mathbb{R}_+^d$ and $\lambda > 0$.
Divide $\mathcal{X}$ into $q^d$ identical subcubes.
If $X_n \subset \mathcal{X}$ is a set of $n \geq q^d$ points such that each subcube contains at least one of these points, then for all $p \in \pi_\ell(\mathbb{R}^d)$,
$$
\|p\|_{L^\infty(\mathcal{X})} \leq e^{2 d \gamma_d (\ell + 1)} \|p\|_{L^\infty(X_n)} .
$$
\end{lemma}

Now let $\mathcal{X}$ be a domain and $X_n = \{\mathbf{x}_1,\dots,\mathbf{x}_n\} \subset \mathcal{X}$.
The elementary tool that we use is the \emph{sampling operator}
\begin{align*}
T_{X_n}: \pi_\ell(\mathcal{X}) & \rightarrow \mathbb{R}^n  \\
p & \mapsto (p(\mathbf{x}_1),\dots,p(\mathbf{x}_n)) \nonumber
\end{align*}
where $\pi_\ell(\mathcal{X})$ is equipped with the norm $\|\cdot\|_{L^\infty(\mathcal{X})}$.
The set of pointwise evaluation functionals on $X_n$ is called a \emph{norming set} for $\pi_\ell(\mathcal{X})$ if the sampling operator $T_{X_n}$ is injective into $\mathbb{R}^n$.
For shorthand, we simply call such $X_n$ a norming set.
The terminology comes from the fact that such a $T_{X_n}$ can be used to induce a norm on $\pi_\ell(\mathcal{X})$ using the norm on $\mathbb{R}^n$.
For further details on sampling operators and norming sets, see Section 4 of \citet{mhaskar2001spherical}.
For a continuous linear function $T : U \rightarrow V$ between normed spaces $U$ and $V$, we let $\|T\|$ denote the operator norm $\sup_{u \neq 0} \|T(u)\|_V / \|u\|_U$.
In what follows, \Cref{prop: norming set,prop: bounded weights,prop: poly repro} provide analogues of Theorems 3.4, 3.8, and 11.21 of \citet{wendland2004scattered} that are adapted to the case of the box fill distance (as opposed to the ball fill distance) to obtain sharper constants for use in the present context.

\begin{proposition}[Properties of the sampling operator] \label{prop: norming set}
Let $\mathcal{X} = [\mathbf{a}, \mathbf{a} + \lambda \mathbf{1}]$ for some $\mathbf{a} \in \mathbb{R}_+^d$ and $\lambda > 0$.
Let $\ell \in \mathbb{N}$.
Suppose that $\rho_{X_n,\mathcal{X}} \leq \lambda / (\gamma_d (\ell + 1))$.
Then $X_n$ is a norming set and $\|T_{X_n}^{-1}\| \leq e^{2d \gamma_d (\ell+1)}$.
\end{proposition}
\begin{proof}
If $T_{X_n}$ is not injective then, since $T_{X_n}$ is linear, there must be a $0 \neq p \in \pi_\ell(\mathcal{X})$ for which $T_{X_n}(p) = 0$.
Let $0 \neq p \in \pi_\ell(\mathbb{R}^d)$.
Our first task is to show that $T_{X_n}(p) \neq 0$.
Since $\rho_{X_n,\mathcal{X}} \leq \lambda / (\gamma_d (\ell + 1))$, the conditions of \Cref{lem: technical} are satisfied with $q = \gamma_d (\ell+1)$.
Thus the conclusion of \Cref{lem: technical} holds, namely that
\begin{align*}
\|T_{X_n}(p)\|_\infty = \|p\|_{L^\infty(X_n)} \geq e^{-2d\gamma_d(\ell+1)} \|p\|_{L^\infty(\mathcal{X})} > 0 .
\end{align*}
Thus $T_{X_n}(p) \neq 0$, showing that $T_{X_n}$ must be injective and $X_n$ must be a norming set.
Finally, we have also shown that
\begin{align*}
\|T_{X_n}^{-1}\| 
= \sup_{p \neq 0} \frac{\|p\|_{L^\infty(\mathcal{X})}}{\|T_{X_n}(p)\|} 
\leq \sup_{p \neq 0} \frac{\|p\|_{L^\infty(\mathcal{X})}}{\|T_{X_n}(p)\|_\infty} 
\leq e^{2d \gamma_d (\ell+1)} ,
\end{align*}
as claimed.
\end{proof}

Using a norming set $X_n$, we can establish a global form of polynomial reproduction, as stated in the next result:

\begin{proposition}[Global polynomial reproduction] \label{prop: bounded weights}
In the setting of \Cref{prop: norming set}, for each $\mathbf{x} \in \mathcal{X}$ there exists $\mathbf{u}(\mathbf{x}) \in \mathbb{R}^n$ such that
\begin{itemize}
\item $p(\mathbf{x}) = \sum_{i=1}^n u_i(\mathbf{x}) p(\mathbf{x}_i)$ for all $p \in \pi_\ell(\mathcal{X})$,
\item $\sum_{i=1}^n |u_i(\mathbf{x})| \leq e^{2d \gamma_d (\ell+1)}$ . 
\end{itemize}
\end{proposition}
\begin{proof}
Define a function $f : \text{range}(T_{X_n}) \rightarrow \mathbb{R}$ by $f(\mathbf{v}) = T_{X_n}^{-1}(\mathbf{v})(\mathbf{x})$.
From the conclusion of \Cref{prop: norming set} we have $\|f\| \leq \|T_{X_n}^{-1}\| \leq e^{2d \gamma_d (\ell+1)}$.
By the Hahn--Banach theorem, $f$ has a norm-preserving extension $f_{\text{ext}}$ to $\mathbb{R}^n$.
Since $f_{\text{ext}}$ is a linear function on $\mathbb{R}^n$, it can be written as $f_{\text{ext}}(\mathbf{v}) = \langle \mathbf{v} , \mathbf{u} \rangle$ from the Riesz representation theorem where $\|\mathbf{u}\| = \|f_{\text{ext}}\| \leq e^{2d \gamma_d (\ell+1)}$.
Then $p(\mathbf{x}) = f(T_{X_n}(p)) = f_{\text{ext}}(T_{X_n}(p)) = \sum_{i=1}^n u_i p(\mathbf{x}_i)$, as required.
\end{proof}

The application of \Cref{prop: bounded weights} to a collection of smaller hypercubes contained in $\mathcal{X}$ yields the following local form of polynomial reproduction that will be used to prove \Cref{prop: mean conver}:

\begin{proposition}[Local polynomial reproduction] \label{prop: poly repro}
Suppose that $\mathcal{X} = [\mathbf{0}, \gamma \mathbf{1}]$.
Let $\ell \in \mathbb{N}$ and $0 < \lambda \leq \gamma$.
Suppose that $\rho_{X_n,\mathcal{X}} \leq \lambda / (\gamma_d (\ell + 1))$.
Then for all $\mathbf{x} \in \mathcal{X}$ there exist numbers $u_i(\mathbf{x})$ with
\begin{itemize}
\item $\sum_{i=1}^n u_i(\mathbf{x}) p(\mathbf{x}_i) = p(\mathbf{x})$ for all $p \in \pi_\ell(\mathcal{X})$,
\item $\sum_{i=1}^n |u_i(\mathbf{x})| \leq e^{2d \gamma_d (\ell+1)}$ , 
\item $u_i(\mathbf{x}) = 0$ if $\|\mathbf{x}-\mathbf{x}_i\|_\infty > \lambda$ .
\end{itemize}
\end{proposition}
\begin{proof}
Given $\mathbf{x} \in \mathcal{X}$ we can find a box $\mathcal{X}_{\mathbf{x}} \coloneqq [\mathbf{a},\mathbf{a} + \lambda \mathbf{1}] \subset \mathcal{X}$ with $\mathbf{x} \in \mathcal{X}_{\mathbf{x}}$.
Let $Y_n := X_n \cap \mathcal{X}_{\mathbf{x}} = \{\mathbf{y}_1,\dots,\mathbf{y}_m\}$.
The box $\mathcal{X}_{\mathbf{x}}$ and the set $Y_n$ satisfy the conditions of \Cref{prop: bounded weights} and thus there exists $\tilde{u}_i(\mathbf{x})$ such that
\begin{itemize}
\item $p(\mathbf{x}) = \sum_{ i = 1}^m \tilde{u}_i(\mathbf{x}) p(\mathbf{y}_i)$ for all $p \in \pi_\ell(\mathbb{R}^d)$ ,
\item $\sum_{ i = 1 }^m |\tilde{u}_i(\mathbf{x})| \leq e^{2d \gamma_d (\ell+1)}$ .
\end{itemize}
From here we define $u_i(\mathbf{x}) = \tilde{u}_j(\mathbf{x})$ if $\mathbf{x}_i = \mathbf{y}_j$, otherwise we define $u_i(\mathbf{x}) = 0$, ensuring that $u_i(\mathbf{x}) = 0$ if $\|\mathbf{x} - \mathbf{x}_i\|_\infty > \lambda$.
\end{proof}

\subsection{Proof of \Cref{prop: mean conver}}
\label{app: proof thm 1}

This section contains the proof of \Cref{prop: mean conver}.
Before we begin, it is useful to recall that $\mathcal{H}_k(\mathcal{X})$ is endowed with the semi-norm $|f|_{\mathcal{H}_k(\mathcal{X})}$ which can be computed using either of the equivalent formulations $\|\mathbf{x} \mapsto f(\mathbf{x}) - f(\mathbf{0})\|_{\mathcal{H}_{k_b}(\mathcal{X})}$ or $\|\mathbf{x} \mapsto (f(\mathbf{x}) - f(\mathbf{0})) / b(\mathbf{x})\|_{\mathcal{H}_{k_e}(\mathcal{X})}$ \citep[][Theorem~5.16]{Paulsen2016}; here the latter formulation will be used. 
Let $\mathcal{S}(X_n)$ denote the set of linear functionals of the form $\mathsf{s}[f] = \mathbf{w}^\top f(X_n)$ whose weights are normalised such that $\mathbf{1}^\top \mathbf{w} = 1$.

\begin{proposition}
\label{prop: wce realised}
The extrapolated mean estimator $\mathsf{s}[f] = m_n[f](\mathbf{0})$ from \eqref{eq: simplified expressions} minimises the worst case error
$$
\mathrm{wce}(\mathsf{s}) := \sup \left\{ | f(\mathbf{0}) - \mathsf{s}[f] |  \; : \; |f|_{\mathcal{H}_k(\mathcal{X})} \leq 1 \right\} 
$$
among all estimators $\mathsf{s} \in \mathcal{S}(X_n)$.
\end{proposition}
\begin{proof}
Consider a general element $\mathsf{s} \in \mathcal{S}(X_n)$, which has the form $\mathsf{s}[f] = \mathbf{w}^\top f(X_n)$ where $\mathbf{1}^\top \mathbf{w} = 1$.
Since the weights $\mathbf{w}$ are normalised, $\mathsf{s}$ is exact on constant functions, and thus $f(\mathbf{0}) - \mathsf{s}[f] = \mathsf{s}[f(\mathbf{0}) - f]$ where $f(\mathbf{0}) - f \in \mathcal{H}_{k_b}(\mathcal{X})$.
It follows that 
$$
\mathrm{wce}(\mathsf{s}) = \sup \big\{ | \mathsf{s}[f] |  \; : \; \|f\|_{\mathcal{H}_{k_b}(\mathcal{X})} \leq 1 \big\} .
$$
It has been assumed that the elements of $X_n$ are distinct.
The Riesz representer of the error functional $f \mapsto f(\mathbf{0}) - \mathbf{w}^\top f(X_n)$ is $k_b(\mathbf{0},\cdot) - \mathbf{w}^\top \mathbf{k}_b(\cdot)$ and, since $\mathrm{wce}(\mathsf{s})$ is the operator norm of the error functional and Hilbert spaces are self-dual, we have that
\begin{align}
\mathrm{wce}(\mathsf{s})^2 = \left\| k_b(\mathbf{0},\cdot) - \mathbf{w}^\top \mathbf{k}_b(\cdot) \right\|_{\mathcal{H}_{k_b}(\mathcal{X})}^2
& = k_b(\mathbf{0},\mathbf{0}) - 2 \mathbf{w}^\top \mathbf{k}_b(\mathbf{0}) + \mathbf{w}^\top \mathbf{K}_b \mathbf{w} \nonumber \\
& = \mathbf{w}^\top \mathbf{K}_b \mathbf{w}  \label{eq: quad form}
\end{align}
where the final equality follows since $k_b(\mathbf{0},\mathbf{x}) = b(\mathbf{0}) b(\mathbf{x}) k_e(\mathbf{0},\mathbf{x}) = 0$ for all $\mathbf{x} \in \mathcal{X}$.
This leaves a quadratic form in $\mathbf{w}$ and we seek the minimum subject to $\mathbf{1}^\top \mathbf{w} = 1$.
If we consider the Lagrangian
\begin{align*}
\mathcal{L}(\mathbf{w},\lambda) = \mathbf{w}^\top \mathbf{K}_b \mathbf{w} + \lambda( \mathbf{1}^\top \mathbf{w} - 1 )
\end{align*}
then we have a critical point when $\partial_{\mathbf{w}} \mathcal{L} = 2 \mathbf{K}_b \mathbf{w} + \lambda \mathbf{1} = 0$, so that $\mathbf{w} = - \frac{\lambda}{2} \mathbf{K}_b^{-1} \mathbf{1}$.
Enforcing the normalisation constraint leads to $\mathbf{w} = \mathbf{K}_b^{-1} \mathbf{1} / (\mathbf{1}^\top \mathbf{K}_b^{-1} \mathbf{1})$, which means that $\mathsf{s}[f] = \mathbf{w}^\top f(X_n) = m_n[f](\mathbf{0})$, as claimed.
\end{proof}

The proof of \Cref{prop: mean conver} below employs the \emph{multi-index} notation, meaning that for $\bm{\beta} \in \mathbb{N}_0^d$ we let $|\bm{\beta}| \coloneqq \beta_1 + \cdots + \beta_d$, $\bm{\beta}! \coloneqq \beta_1 \cdots \beta_d$, and $\mathbf{z}^{\bm{\beta}} \coloneqq z_1^{\beta_1} \cdots z_d^{\beta_d}$ for $\mathbf{z} \in \mathbb{R}^d$.

\begin{proof}[Proof of \Cref{prop: mean conver}]
Assume without loss of generality that $|f|_{\mathcal{H}_k(\mathcal{X})} \neq 0$, since otherwise $m_n^h[f] = 0$ is trivially exact.
Let $\mathbf{x}_i^h = h \mathbf{x}_i$ and $k_b(\mathbf{x},\mathbf{y}) = b(\mathbf{x}) b(\mathbf{y}) k_e(\mathbf{x},\mathbf{y})$ in shorthand.
Introduce the positive semi-definite quadratic form
\begin{align}
Q(\mathbf{u}) & = \sum_{i=1}^n \sum_{j=1}^n u_i u_j k_b(\mathbf{x}_i^h,\mathbf{x}_j^h) \label{eq:  Q fn}
\end{align}
and note from \Cref{prop: wce realised} that
\begin{align} 
\frac{|f(\mathbf{0}) - m_n^h[f](\mathbf{0})|}{|f|_{\mathcal{H}(k)}} \leq \mathrm{wce}(m_n^h[f]) = \min\left\{ Q(\mathbf{u})^{1/2} : \mathbf{u} \in \mathbb{R}^n, \; \sum_{i=1}^n u_i = 1 \right\} . \label{eq: Qn bound}
\end{align}
Thus we can bound the relative error by the square root of $Q(\mathbf{u})$ for any choice of $\mathbf{u} \in \mathbb{R}^n$ for which $\sum_{i=1}^n u_i = 1$.
Our choice of $\mathbf{u}$ will be based on polynomial reproduction, as described next.

From Taylor's theorem, since $k_e(\mathbf{x},\cdot) \in C^{2s}(\mathcal{X}_h)$,
\begin{align}
k_e(\mathbf{x},\mathbf{y}) = \sum_{|\bm{\beta}| < 2s} \frac{\partial_{\mathbf{z}}^{\bm{\beta}} k_e(\mathbf{x},\mathbf{z})|_{\mathbf{z} = \mathbf{x}}}{\bm{\beta}!} (\mathbf{y}-\mathbf{x})^{\bm{\beta}} + R(\mathbf{x},\mathbf{y}) , \qquad | R(\mathbf{x},\mathbf{y}) | \leq C_{\mathbf{x}}^{(2s)} \|\mathbf{x}-\mathbf{y}\|^{2s } \label{eq: Taylor}
\end{align}
for all $\mathbf{x},\mathbf{y} \in \mathcal{X}_h$, where the constants $C_{\mathbf{x}}^{(2s)} \coloneqq ((2s)!)^{-1} \sup_{\mathbf{z} \in \mathcal{X}_h} \sum_{|\bm{\beta}| = 2s} \partial_{\mathbf{z}}^{\bm{\beta}} k_e(\mathbf{x},\mathbf{z})$ are uniformly bounded by $C_k^{(2s)} \coloneqq \sup_{\mathbf{x} \in \mathcal{X}} C_{\mathbf{x}}^{(2s)}$ since $\mathcal{X}$ is compact.
Plugging this into \eqref{eq:  Q fn} gives that
\begin{align*}
Q(\mathbf{u}) & = \sum_{i=1}^n u_i b(\mathbf{x}_i^h) \sum_{j=1}^n u_j b(\mathbf{x}_j^h) \left(  \sum_{|\bm{\beta}| < 2s} \frac{\partial_{\mathbf{z}}^{\bm{\beta}} k_e(\mathbf{x}_i^h,\mathbf{z})|_{\mathbf{z} = \mathbf{x}_i^h}}{\bm{\beta}!} (\mathbf{x}_j^h-\mathbf{x}_i^h)^{\bm{\beta}} + R(\mathbf{x}_i^h,\mathbf{x}_j^h) \right) .
\end{align*}
Since $b(\mathbf{x}_j^h) (\mathbf{x}_j^h - \mathbf{x}_i^h)^{\bm{\beta}}$ is a polynomial in $\mathbf{x}_j^h$ of total degree at most $r + |\bm{\beta}|$ and $|\bm{\beta}| < 2s$, we aim to pick a vector $\mathbf{u}$ for which local (at $\mathbf{0}$) polynomial reproduction on $X_n^h$ occurs up to polynomials of total order $\ell \coloneqq r + (2s-1)$.
Set $\lambda  \coloneqq \gamma_d (\ell + 1) \rho_{X_n^h,\mathcal{X}_h} = \gamma_d (\ell + 1) h \rho_{X_n,\mathcal{X}}$, so that our assumption on the box fill distance implies $0 < \lambda \leq h$, and trivially $\rho_{X_n^h,\mathcal{X}_h} \leq \lambda / (\gamma_d (\ell + 1))$. 
Thus the conditions of \Cref{prop: poly repro} are satisfied, and there exists $\mathbf{u} \in \mathbb{R}^n$ such that
\begin{itemize}
\item $\sum_{i=1}^n u_i p(\mathbf{x}_i^h) = p(\mathbf{0})$ for all $p \in \pi_\ell(\mathcal{X}_h)$ ,
\item $\sum_{i=1}^n |u_i| \leq e^{2d \gamma_d (\ell+1)}$ ,
\item $u_i = 0$ if $\mathbf{x}_i^h \notin [\mathbf{0}, \lambda \mathbf{1}]$ .
\end{itemize}
For this choice of $\mathbf{u}$ it follows from the local polynomial reproduction property that, for $|\bm{\beta}| < 2s$,
\begin{align*}
\sum_{j=1}^n u_j b(\mathbf{x}_j^h) (\mathbf{x}_j^h - \mathbf{x}_i^h)^{\bm{\beta}} = \left. b(\mathbf{x}) (\mathbf{x} - \mathbf{x}_i^h)^{\bm{\beta}} \right|_{\mathbf{x} = \mathbf{0}} = b(\mathbf{0}) (-\mathbf{x}_i^h)^{\bm{\beta}} = \mathbf{0}
\end{align*}
since $b(\mathbf{0}) = 0$, and thus, recalling that $u_i = 0$ whenever $\mathbf{x}_i^h \notin [\mathbf{0}, \lambda \mathbf{1}]$,
\begin{align*}
Q(\mathbf{u}) & = \sum_{i=1}^n u_i b(\mathbf{x}_i^h) \sum_{j=1}^n u_j b(\mathbf{x}_j^h) R(\mathbf{x}_i^h,\mathbf{x}_j^h) 
\leq \|\mathbf{u}\|_1^2  \|b\|_{L^\infty(\mathcal{X}_h)}^2 \sup_{\mathbf{x},\mathbf{y} \in [\mathbf{0},\gamma \mathbf{1}]} |R(\mathbf{x},\mathbf{y})| ,
\end{align*}
where by construction $\|\mathbf{u}\|_1 \leq e^{2d \gamma_d (\ell+1)}$. 
The final term can be bounded using the error estimate in \eqref{eq: Taylor}:
\begin{align*}
\sup_{\mathbf{x},\mathbf{y} \in [\mathbf{0},\gamma \mathbf{1}]} |R(\mathbf{x},\mathbf{y})|
\leq \sup_{\mathbf{x} \in [\mathbf{0},\gamma\mathbf{1}]} C_{\mathbf{x}}^{(2s)} \times \|\mathbf{0} - \lambda \mathbf{1}\|^{2s} 
&\leq C_k^{(2s)} \frac{d^s}{(2s)!} \lambda^{2s} \\
&= C_k^{(2s)} \frac{d^s}{(2s)!} ( \gamma_d (\ell+1) h \rho_{X_n,\mathcal{X}} )^{2s} , 
\end{align*}
leading to the claimed overall bound $\mathrm{wce}(m_n^h[f]) \leq C_{r,s}  h^s \rho_{X_n,\mathcal{X}}^s \|b\|_{L^\infty(\mathcal{X}_h)}$ on the relative error, where $C_{r,s} \coloneqq (C_k^{(2s)})^{1/2} d^{s/2} e^{2 d \gamma_d (\ell + 1)} \gamma_d^s (\ell + 1)^s / \sqrt{(2s)!}$. 
\end{proof}

\subsection{Proof of \Cref{cor: spectral}}
\label{app: proof of spectral}

\begin{proof}[Proof of \Cref{cor: spectral}]
The assumption on the growth of the derivatives of $k_e$ implies that $C_k^{(2s)} \leq C_k^{2s}$ for the assumed constant $C_k$, and we employ this bound throughout.
Now, to more easily track the $s$-dependent constants in the bound of \Cref{prop: mean conver} we employ a simpler upper bound
\begin{align}
C_{r,s} h^s \rho_{X_n,\mathcal{X}}^s & = \frac{ C_k^s d^{s/2} e^{2 d \gamma_d (r + 2s)} \gamma_d^s (r + 2s)^s }{ \sqrt{(2s)!} }  h^s \rho_{X_n,\mathcal{X}}^s 
\leq e^{2d\gamma_d r} \left( \frac{ C_k d^{1/2} e^{4 d \gamma_d + 1} h }{ 2s } \right)^s  \label{eq: simpler bound}
\end{align} 
which holds whenever $\rho_{X_n,\mathcal{X}} \leq 1 / (\gamma_d(r + 2s))$, and where we used the elementary fact $1 / \sqrt{(2s)!} \leq (e/(2s))^s$ to obtain this simpler bound.

The idea of this proof is to pick a particular value of $s \in \mathbb{N}_0$ that will be $\rho_{X_n,\mathcal{X}}$-dependent.
For \eqref{eq: simpler bound} to hold we require $\rho_{X_n,\mathcal{X}} \leq 1 / (\gamma_d (r + 2s))$, and the largest such $s$ for which this requirement is satisfied is
\begin{align}
s^\star \coloneqq \left\lfloor - \frac{r}{2} + \frac{1}{2 \gamma_d \rho_{X_n,\mathcal{X}}} \right\rfloor . \label{eq: largest s}
\end{align}
The assumption $\rho_{X_n,\mathcal{X}} \leq 1 / (2 \gamma_d (r+1))$ implies $s^\star \geq 1 / (4 \gamma_d \rho_{X_n,\mathcal{X}}) \geq 0$, so $s$ is well-defined as an element of $\mathbb{N}_0$.
Further, the assumption $\rho_{X_n,\mathcal{X}} \leq 1 / (2 d^{1/2} \gamma_d e^{4 d \gamma_d + 1} )$ implies $t \mapsto [ ( \frac{1}{2} d^{1/2} e^{4 d \gamma_d + 1} h) / t ]^t$ is decreasing on $t \in [1/(4 \gamma_d \rho_{X_n,\mathcal{X}}),\infty)$ for any $h \in (0,1]$.
Thus we can replace $s$ by $1 / (4 \gamma_d \rho_{X_n,\mathcal{X}})$ in \eqref{eq: simpler bound} to obtain the upper bound
\begin{align*}
C_{r,s} h^s \rho_{X_n,\mathcal{X}}^s & \leq e^{2 d \gamma_d r} \left( \frac{ \frac{1}{2} C_k d^{1/2} e^{4 d \gamma_d + 1} h }{ 1 / (4 \gamma_d \rho_{X_n,\mathcal{X}})} \right)^{1 / (4 \gamma_d \rho_{X_n,\mathcal{X}})} 
= C_{n,r,s} h^{ \frac{1}{ 4 \gamma_d \rho_{X_n,\mathcal{X}} }  }
\end{align*}
where $C_{n,r,s} \coloneqq (2 C_k d^{1/2} \gamma_d e^{4d\gamma_d + 1} \rho_{X_n,\mathcal{X}})^{1 / (4 \gamma_d \rho_{X_n,\mathcal{X}})}$ is a $h$-independent constant.
The assumptions we made above on the box fill distance are satisfied when $\rho_{X_n,\mathcal{X}} \leq \min \{ 1 / (2 \gamma_d (r+1)) , 1 / (2 d^{1/2} \gamma_d e^{4 d \gamma_d + 1} ) \}$.
\end{proof}

\subsection{Verifying the Assumptions for \Cref{ex: finite diff}}
\label{app: check assum}

The aim of this appendix is to verify that the function $\psi(x) = \sin(10x) + 1_{x > 0} x^{s+4}$ satisfies $\psi(x) = c_0 + c_1x + c_2x^2 + c_3(x) x^3 $ for some $c_0,c_1,c_2 \in \mathbb{R}$ and $c_3 \in H^{s+1}(\mathcal{O})$ where $\mathcal{O} = (-\delta,\delta)$, $\delta > 0$, is an open neighbourhood of 0 and we recall that $H^{s+1}(\mathcal{O})$ is the Sobolev space of $s+1$ times weakly differentiable functions on $\mathcal{O}$; see \Cref{app: kernels}.
It then follows from \Cref{app: kernels} that $x \mapsto c_3(x)$ and $x \mapsto c_3(-x)$ are elements of $\mathcal{H}_{k_e}(\mathcal{X})$ whenever $k_e$ is either the M\'{a}tern or Wendland kernel with smoothness $s$ and $\mathcal{X} = [0,\delta)$, since these kernels reproduce (up to an equivalent norm) $H^{s+1}(\mathcal{X})$ in dimension $d=1$.

From linearity, it suffices to establish this fact separately for $\psi_1(x) \coloneqq \sin(10x)$ and $\psi_2(x) \coloneqq 1_{x>0} x^{s+4}$.
For the first term, since the trigonometric functions are real-analytic, we have a convergent power series $\psi_1(x) = \sum_{i=0}^\infty \tilde{c}_i x^i$ for $x \in [-\epsilon,\epsilon]$ for some $\epsilon > 0$.
Thus $\psi_1(x) = \tilde{c}_0 + \tilde{c}_1 x + \tilde{c}_2 x^2 + c_3(x) x^3$ with $c_3(x) = \sum_{i=3}^\infty \tilde{c}_i x^{i-3}$, and our task is to show that this latter series is convergent in a neighbourhood of $x = 0$.
To this end, we restrict attention to $x \in [-\epsilon/2,\epsilon/2]$, and then use the ratio test 
$$
\frac{|\tilde{c}_i x^{i-3}|}{|\tilde{c}_i \epsilon^i|} = \frac{1}{|x|^3} \left|\frac{x}{\epsilon}\right|^i \leq \frac{1}{|x|^3} \frac{1}{2^i} \rightarrow 0 \quad \text{as} \quad i \rightarrow \infty
$$
to deduce absolute convergence of the series, as required.
For the second term, we take $c_0 = c_1 = c_2 = 0$, so that $\psi_2(x) = c_3(x) x^3$ where $c_3(x) = 1_{x > 0} x^{s+1}$, observing that $c_3$ has an $(s+1)$-order weak derivative $c_3^{(s+1)}(x)$ taking the value $0$ on $x \leq 0$ and $(s+1)!$ on $x > 0$, so that $c_3 \in H^{s+1}(\mathcal{O})$.

\subsection{Proof of \Cref{prop: discrete} and \Cref{cor: second order expand}}

\label{subsec: proof of continuous extend results}

\begin{proof}[Proof of \Cref{prop: discrete}]
Our assumption implies $(e_n)_{n \in \mathbb{N}}$ is Cauchy and thus this sequence converges to a limit, denoted $e_\infty \in \mathbb{R}$.
Let $\varphi : [0,\infty) \rightarrow [0,1]$ be a smooth function with $\varphi(0) = 0$, $\varphi = 1$ on $[1,\infty)$, and derivatives uniformly bounded.
Let 
\begin{align*}
  e(\mathbf{x}) \coloneqq
  \begin{cases}
    e_\infty &\text{if } \quad \mathbf{x} = \mathbf{0}, \\ e_{n+1} + ( e_n - e_{n+1} ) \prod_{i=1}^d \varphi\left( \scaleto{ \frac{ x_i - \mathbf{x}_{n+1,i}}{\mathbf{x}_{n,i} - \mathbf{x}_{n+1,i}} }{20pt} \right)  &\text{if }\quad \mathbf{x} \in [\mathbf{0},\mathbf{x}_n] \setminus [\mathbf{0},\mathbf{x}_n], \\ e_1 &\text{if }\quad \mathbf{x} \notin [\mathbf{0},\mathbf{x}_1]
  \end{cases}
\end{align*}
so that $e(\mathbf{x}_n) = e_n$ and $e(\mathbf{0}) = e_\infty$.
Now set $f(\mathbf{x}) = y_\infty + b(\mathbf{x}) e(\mathbf{x})$.
To establish $|f|_{\mathcal{H}_k(\mathcal{X})} < \infty$ it suffices to establish that $\lim_{\mathbf{x} \rightarrow \mathbf{0}} \partial^{\bm{\beta}} e(\mathbf{x})$ exists and is finite for all $|\bm{\beta}| = p$, since then we will have $e \in C^p(\mathcal{X}) \subset \mathcal{H}(k_e)$.
This is indeed the case, since on $[\mathbf{0},\mathbf{x}_n] \setminus [\mathbf{0},\mathbf{x}_n]$ we have that
$$
\partial^{\bm{\beta}} e(\mathbf{x}) = (e_n - e_{n+1}) \prod_{i=1}^d \varphi^{(\beta_i)} \left( \frac{ x_i - \mathbf{x}_{n+1,i}}{\mathbf{x}_{n,i} - \mathbf{x}_{n+1,i}} \right) \frac{1}{(\mathbf{x}_{n,i} - \mathbf{x}_{n+1,i})^{\beta_i}} .
$$
where the $\varphi^{(\beta_i)}$ terms are uniformly bounded and the remaining terms vanish since
$$
|e_n - e_{n+1}| \prod_{i=1}^d \frac{1}{(\mathbf{x}_{n,i} - \mathbf{x}_{n+1,i})^{\beta_i}} \leq \frac{|e_n - e_{n+1}|}{ \min(\mathbf{x}_n - \mathbf{x}_{n+1})^p }
$$
where $p = \beta_1 + \dots + \beta_d$, and this final bound was assumed to vanish.
\end{proof}

\begin{proof}[Proof of \Cref{cor: second order expand}]
    The renormalised error is $e_n = C_1 + C_2 x_n^{p+1} + O(x_n^{p+2})$, so that 
    $$
    \frac{e_n - e_{n+1}}{(x_n - x_{n+1})^p} = C_2 \frac{ (x_n^{p+1} - x_{n+1}^{p+1})}{(x_n - x_{n+1})^p} + \frac{O(x_n^{p+2}) + O(x_{n+1}^{p+2})}{ (x_n - x_{n+1})^p} \rightarrow 0 .
    $$
    Our assumptions on the sequence $(x_n)_{n \in \mathbb{N}}$ imply that both terms vanish in the $n \rightarrow \infty$ limit.
\end{proof}

\subsection{Proof of \Cref{prop: not over-confident}}
\label{app: uq proof}


First we establish the correctness of the algebraic identity \eqref{eq: sigma estimator} given in the main text.
Again, it is useful to recall that $\mathcal{H}_k(\mathcal{X})$ is endowed with the semi-norm $|f|_{\mathcal{H}_k(\mathcal{X})}$ which can be computed using either of the equivalent formulations $\|\mathbf{x} \mapsto f(\mathbf{x}) - f(\mathbf{0})\|_{\mathcal{H}_{k_b}(\mathcal{X})}$ or $\|\mathbf{x} \mapsto (f(\mathbf{x}) - f(\mathbf{0})) / b(\mathbf{x})\|_{\mathcal{H}_{k_e}(\mathcal{X})}$ \citep[][Theorem~5.16]{Paulsen2016}; here the latter formulation will be used. 

\begin{proposition} \label{prop: sigma is norm}
The estimator $\sigma_n^2[f] := \frac{1}{n} |m_n[f]|_{\mathcal{H}_k(\mathcal{X})}^2$ has the explicit form
\begin{align*} 
\sigma_n^2[f] 
= \frac{1}{n} \left[ f(X_n)^\top \mathbf{K}_b^{-1} f(X_n) -  \frac{ ( \mathbf{1}^\top \mathbf{K}_b^{-1} f(X_n) )^2 }{ \mathbf{1}^\top \mathbf{K}_b^{-1} \mathbf{1} }  \right]
\end{align*}
given in \eqref{eq: sigma estimator} of the main text.
\end{proposition}
\begin{proof}
Let $\bar{m}_n[f] = m_n[f] - m_n[f](\mathbf{0})$.
From direct calculation,
\begin{align*}
|m_n[f]|_{\mathcal{H}_k(\mathcal{X})}^2 &= \|\bar{m}_n[f]\|_{\mathcal{H}_{k_b}(\mathcal{X})}^2 \\
  & = \left\| \mathbf{k}_b(\cdot)^\top \mathbf{K}_b^{-1} \left\{ f(X_n) -  \left( \frac{\mathbf{1}^\top \mathbf{K}_b^{-1} f(X_n) }{ \mathbf{1}^\top \mathbf{K}_b^{-1} \mathbf{1} } \right) \mathbf{1} \right\} \right\|_{\mathcal{H}_{k_b}(\mathcal{X})}^2 \nonumber \\
&= \left\{ f(X_n) -  \left( \frac{\mathbf{1}^\top \mathbf{K}_b^{-1} f(X_n) }{ \mathbf{1}^\top \mathbf{K}_b^{-1} \mathbf{1} } \right) \mathbf{1} \right\}^\top \mathbf{K}_b^{-1} \left\{ f(X_n) -  \left( \frac{\mathbf{1}^\top \mathbf{K}_b^{-1} f(X_n) }{ \mathbf{1}^\top \mathbf{K}_b^{-1} \mathbf{1} } \right) \mathbf{1} \right\} \nonumber \\
& = f(X_n)^\top \mathbf{K}_b^{-1} f(X_n) -  \frac{ ( \mathbf{1}^\top \mathbf{K}_b^{-1} f(X_n) )^2 }{ \mathbf{1}^\top \mathbf{K}_b^{-1} \mathbf{1} } ,
\end{align*}
which completes the proof.
\end{proof}

Next we present a general result, which forms the crux of the argument in our proof of \Cref{prop: not over-confident}:

\begin{proposition} \label{prop: UQ}
Assume that $f \in \mathcal{H}_k(\mathcal{X})$.
Then
\begin{align*}
\frac{ |f(\mathbf{0}) - m_n[f](\mathbf{0})| }{ \sqrt{k_n[f](\mathbf{0},\mathbf{0})} }  & \leq \frac{ |f|_{\mathcal{H}_k(\mathcal{X})} }{ \sigma_n[f] } .
\end{align*}
\end{proposition}
\begin{proof}
From \eqref{eq: quad form} in the the proof of \Cref{prop: wce realised}, we have $\mathrm{wce}(\mathsf{s})^2 = \mathbf{w}^\top \mathbf{K}_b \mathbf{w}$ for any algorithm $\mathsf{s}[f] = \mathbf{w}^\top f(X_n)$ with $\mathbf{1}^\top \mathbf{w} = 1$.
In particular, we can consider the algorithm $s[f] = m_n[f](\mathbf{0})$, which has $\mathbf{w} = \mathbf{K}_b^{-1} \mathbf{1} / (\mathbf{1}^\top \mathbf{K}_n^{-1} \mathbf{1})$, to see that 
\begin{align*}
\mathrm{wce}(\mathsf{s})^2 = \left( \frac{\mathbf{K}_b^{-1} \mathbf{1}}{\mathbf{1}^\top \mathbf{K}_b^{-1} \mathbf{1}} \right)^\top \mathbf{K}_b \left( \frac{\mathbf{K}_b^{-1} \mathbf{1}}{\mathbf{1}^\top \mathbf{K}_b^{-1} \mathbf{1}} \right) = \frac{1}{\mathbf{1}^\top \mathbf{K}_b^{-1} \mathbf{1}} = \frac{k_n[f](\mathbf{0},\mathbf{0})}{\sigma_n^2[f]} .
\end{align*}
This implies that
\begin{align*}
\frac{ |f(\mathbf{0}) - m_n[f](\mathbf{0})| }{ \sqrt{k_n[f](\mathbf{0},\mathbf{0})} }  & \leq \frac{\mathrm{wce}(\mathsf{s}) |f|_{\mathcal{H}_k(\mathcal{X})}}{ \sqrt{k_n[f](\mathbf{0},\mathbf{0})} } \leq \frac{ |f|_{\mathcal{H}_k(\mathcal{X})} }{ \sigma_n[f] } ,
\end{align*}
from which the result is established.
\end{proof}

\begin{proof}[Proof of \Cref{prop: not over-confident}]
   
  For $f \in \mathcal{H}_k(\mathcal{X})$ we have the decomposition $f(\mathbf{x}) = f(\mathbf{0}) + \bar{f}(\mathbf{x})$ with $\bar{f} \in \mathcal{H}_{k_b}(\mathcal{X})$, and $|f|_{\mathcal{H}_k(\mathcal{X})} = \|\bar{f}\|_{\mathcal{H}_{k_b}(\mathcal{X})}$.
  From \Cref{prop: UQ} and $\sigma_n^2[f] = \frac{1}{n} |m_n[f]|_{\mathcal{H}_k(\mathcal{X})}^2$, it follows that
  \begin{equation*}
    \frac{\lvert f(\mathbf{0}) - m_n^h[f](\mathbf{0}) \rvert}{\sqrt{k_n^h[f](\mathbf{0}, \mathbf{0})}} \leq \frac{ n \lvert f \rvert_{\mathcal{H}_{k}(\mathcal{X})} }{ \lvert m_n^h[f] \rvert_{\mathcal{H}_{k}(\mathcal{X})}} = \frac{ n \lVert \bar{f} \rVert_{\mathcal{H}_{k_b}(\mathcal{X})} }{ \lVert \bar{m}_n^h[f] \rVert_{\mathcal{H}_{k_b}(\mathcal{X})}},
  \end{equation*}
  where similarly we have decomposed $m_n^h[f](\mathbf{x}) = m_n^h[f](\mathbf{0}) + \bar{m}_n^h[f](\mathbf{x})$.
  Since $\bar{f}$ and $\bar{m}_n^h[f]$ are elements of $\mathcal{H}_{k_b}(\mathcal{X})$ and $k_b(\mathbf{x}, \mathbf{x}') = b(\mathbf{x}) b(\mathbf{x}') k_e(\mathbf{x}, \mathbf{x}')$, we may write
  \begin{equation*}
    \bar{f}(\mathbf{x}) = b(\mathbf{x}) e(\mathbf{x}) \quad \text{ and } \quad \bar{m}_n^h[f](\mathbf{x}) = b(\mathbf{x}) \tilde{e}_h(\mathbf{x}) 
  \end{equation*}
  for certain $e,\tilde{e}_h \in \mathcal{H}_{k_e}(\mathcal{X})$, with $\| \bar{f} \|_{\mathcal{H}_{k_b}(\mathcal{X})} = \| e \|_{\mathcal{H}_{k_e}(\mathcal{X})}$ and $\| \bar{m}_n^h[f] \|_{\mathcal{H}_{k_b}(\mathcal{X})} = \| \tilde{e}_h \|_{\mathcal{H}_{k_e}(\mathcal{X})}$.
  Therefore
  \begin{equation} \label{eq:UQ-ratio-bound-e}
    \frac{\lvert f(\mathbf{0}) - m_n^h[f](\mathbf{0}) \rvert}{\sqrt{k_n^h[f](\mathbf{0}, \mathbf{0})}} \leq \frac{ n \lVert e \lVert_{\mathcal{H}_{k_e}(\mathcal{X})} }{ \lVert \tilde{e}_h \rVert_{\mathcal{H}_{k_e}(\mathcal{X})} }.
  \end{equation}  
  By construction, $m_n^h[f](h \mathbf{x}_i) = f(h \mathbf{x}_i)$ for $i = 1, \ldots, n$.
  Because $f(\mathbf{x}) = f(\mathbf{0}) + b(\mathbf{x}) e(\mathbf{x})$, the function $\tilde{e}_h$ satisfies
  \begin{equation*}
    \tilde{e}_h(h \mathbf{x}_i) = e(h \mathbf{x}_i) + b(h \mathbf{x}_i)^{-1}[ f(\mathbf{0}) - m_n^h[f](\mathbf{0})]
  \end{equation*}
  for each $i = 1, \ldots, n$.
  The assumptions of \Cref{prop: mean conver} with $s = 1$ are satisfied, and thus
  \begin{equation*}
    b(h \mathbf{x}_i)^{-1} \lvert f(\mathbf{0}) - m_n^h[f](\mathbf{0}) \rvert \leq C_{r, 1} h \rho_{X_n,\mathcal{X}} \frac{ \|b\|_{L^\infty(\mathcal{X}_h)}}{b(h \mathbf{x}_i)} |f|_{\mathcal{H}_k(\mathcal{X})} \to 0
  \end{equation*}
  as $h \to 0$ since, $b$ being a polynomial, $\lim_{h \to 0} \|b\|_{L^\infty(\mathcal{X}_h)} b(h \mathbf{x}_i)^{-1} < \infty$.
  The continuity of $e$, which follows from $e \in \mathcal{H}_{k_e}(\mathcal{X})$ and the continuity of $k_e$, then implies that $\tilde{e}_h(h \mathbf{x}_i) \to e(\mathbf{0})$ as $h \to 0$.
  Let $c \in \mathbb{R}$. The function $g(\mathbf{x}) := c \, k_e(\mathbf{x}, \mathbf{x}') k_e(\mathbf{x}', \mathbf{x}')^{-1}$ has minimal norm among all functions in $\mathcal{H}_{k_e}(\mathcal{X})$ that equal $c$ at $\mathbf{x}'$~\citep[Corollary~3.5]{Paulsen2016}.
  By the reproducing property $\langle k_e(\cdot, \mathbf{x}_1), k_e(\cdot, \mathbf{x}_2) \rangle_{\mathcal{H}_{k_e}(\mathcal{X})}= k_e(\mathbf{x}_1, \mathbf{x}_2)$ in the RKHS inner product, from which it follows that the norm of this function is
  \begin{equation*}
    \lVert g \rVert_{\mathcal{H}_{k_e}(\mathcal{X})} = \langle g , g \rangle_{\mathcal{H}_{k_e}(\mathcal{X})}^{1/2} = \lvert c \rvert \, k_e(\mathbf{x}', \mathbf{x}')^{1/2}.
  \end{equation*}
  Consequently, setting $\mathbf{x}' = h \mathbf{x}_1$ and $c = \tilde{e}_h(h\mathbf{x}_1)$ gives
  \begin{equation*}
    \lVert \tilde{e}_h \rVert_{\mathcal{H}_{k_e}(\mathcal{X})} \geq \lvert \tilde{e}_h(h\mathbf{x}_1) \rvert k_e(h \mathbf{x}_1, h \mathbf{x}_1)^{-1/2} \to \lvert e(\mathbf{0}) \rvert k_e(\mathbf{0}, \mathbf{0})^{-1/2}.
  \end{equation*}
  Using this bound in~\eqref{eq:UQ-ratio-bound-e} then shows that the ratio on the left-hand side is bounded as $h \to 0$, since we have assumed $e(\mathbf{0}) \neq 0$.
\end{proof}

\subsection{Proof of \Cref{prop: kernel parameter estimation}}
\label{app: quasi likelihood sec}

\begin{proof}[Proof of \Cref{prop: kernel parameter estimation}]
Recall that the estimator $\mathbf{r}_n^h[f]$ is any maximiser of
\begin{align*}
  \mathcal{L}_n^h(\mathbf{r}) & = - \underbrace{ \left\{ f(X_n^h)^\top \mathbf{K}_{b_{\mathbf{r}},h}^{-1} f(X_n^h) -  \frac{ ( \mathbf{1}^\top \mathbf{K}_{b_{\mathbf{r}},h}^{-1} f(X_n^h) )^2 }{ \mathbf{1}^\top \mathbf{K}_{b_{\mathbf{r}},h}^{-1} \mathbf{1} } \right\} }_{ =: Q_n^h(f, \mathbf{r})} - \log \det \mathbf{K}_{b_{\mathbf{r}},h} .
\end{align*}
Maximising $\mathcal{L}_n^h(\mathbf{r})$ is equivalent to maximising
\begin{align*}
  \mathcal{J}_n^h(\mathbf{r}) = \mathcal{L}_n^h(\mathbf{r}) - \mathcal{L}_n^h(\mathbf{r}_0).
\end{align*}
Let $\bar{f}(\mathbf{x}) := f(\mathbf{x}) - f(\mathbf{0})$.
By assumption, $\bar{f}$ is an element of the RKHS of the covariance function $k_{b_{\mathbf{r}_0}}$.
According to \Cref{prop: sigma is norm}, $Q_n^h(f, \mathbf{r}) = |m_n^h[f]|_{\mathcal{H}_k(\mathcal{X})}^2 \geq 0$.
Also from \Cref{prop: sigma is norm} and the minimal norm characterisation of the interpolant, we obtain the $h$-independent bound $Q_n^h(f, \mathbf{r}_0) \leq \lVert \bar{f} \rVert_{\mathcal{H}_{k_{b_{\mathbf{r}_0}}}(\mathcal{X})}^2$.
Therefore
\begin{align*}
  \mathcal{J}_n^h(\mathbf{r}) &= Q_n^h(f, \mathbf{r}_0) + \log \det \mathbf{K}_{b_{\mathbf{r}_0},h} - Q_n^h(f, \mathbf{r}) - \log \det \mathbf{K}_{b_{\mathbf{r}},h} \\
  &\leq \log \det \mathbf{K}_{b_{\mathbf{r}_0},h} - \log \det \mathbf{K}_{b_{\mathbf{r}},h} + \lVert \bar{f} \rVert_{\mathcal{H}_{k_{b_{\mathbf{r}_0}}}(\mathcal{X})}^2.
\end{align*}
Because the covariance matrix factorises as $\mathbf{K}_{b_{\mathbf{r}},h} = \mathbf{B}_{\mathbf{r},h} \mathbf{K}_{e, h} \mathbf{B}_{\mathbf{r},h}$, where $\mathbf{K}_{e,h}$ is the covariance matrix for $k_e$ at $X_n^h$ and $\mathbf{B}_{\mathbf{r},h}$ is a diagonal matrix with entries $b_{\mathbf{r}}(h \mathbf{x}_i)$, we get
\begin{align*}
  \mathcal{J}_n^h(\mathbf{r}) \leq 2 \sum_{i=1}^n \log \frac{b_{\mathbf{r}_0}(h \mathbf{x}_i)}{b_{\mathbf{r}}(h \mathbf{x}_i)}  + \lVert \bar{f} \rVert_{\mathcal{H}_{k_{b_{\mathbf{r}_0}}}(\mathcal{X})}^2 .
\end{align*}
Then, since the class of error bounds $b_{\mathbf{r}}$ is monotonically parametrised, $\mathcal{J}_n^h(\mathbf{r}) \to - \infty$ as $h \to 0$ uniformly over $\mathbf{r} \in [\mathbf{0}, \mathbf{r}_0 - \bm{\epsilon}]$ for any $\bm{\epsilon} > \mathbf{0}$.
Because $\mathcal{J}_n^h(\mathbf{r}_0) = 0$, this establishes that $\liminf_{h \to 0} \mathbf{r}_n^h[f] \geq \mathbf{r}_0$, as claimed.
\end{proof}

\subsection{Calculations for Multidimensional Output}
\label{app: multivar calcs}

To simplify the presentation we assume throughout that $\sigma = 1$, since $\sigma$ enters only as a multiplicative constant that can be propagated through the calculations at the end.
Let $k_{\mathcal{X}}(\mathbf{x},\mathbf{x}') := k_0^2 + b(\mathbf{x}) b(\mathbf{x}') k_e(\mathbf{x},\mathbf{x}')$.
From the Kronecker decomposition $\mathbf{K} = \mathbf{K}_{\mathcal{X}} \otimes \mathbf{K}_{\mathcal{T}}$ we have that
\begin{align*}
    \mathbf{K}^{-1} = (\mathbf{K}_{\mathcal{X}}^{-1}) \otimes (\mathbf{K}_{\mathcal{T}}^{-1})
\end{align*}
where from \eqref{eq: after Woodbury} we know that 
\begin{align*}
\mathbf{K}_{\mathcal{X}}^{-1} = \mathbf{K}_b^{-1} - \mathbf{K}_b^{-1} \mathbf{1} ( k_0^{-2} + \mathbf{1}^\top \mathbf{K}_b^{-1} \mathbf{1} )^{-1} \mathbf{1}^\top \mathbf{K}_b^{-1} .
\end{align*}
Let $\mathbf{k}(\mathbf{x},\mathbf{t})$ be the column vector with entries $k((\mathbf{x}_i,\mathbf{t}_i),(\mathbf{x},\mathbf{t}))$, and analogously define $\mathbf{k}_{\mathcal{X}}(\mathbf{x})$ and $\mathbf{k}_{\mathcal{T}}(\mathbf{t})$ as the column vectors with respective entries $k_{\mathcal{X}}(\mathbf{x}_i,\mathbf{x})$ and $k_{\mathcal{T}}(\mathbf{t}_i,\mathbf{t})$.
In this notation we have also the Kronecker decomposition $\mathbf{k}(\mathbf{x},\mathbf{t}) = \mathbf{k}_{\mathcal{X}}(\mathbf{x}) \otimes \mathbf{k}_{\mathcal{T}}(\mathbf{t})$ and the algebraic result that $\mathbf{k}(\mathbf{x},\mathbf{t})^\top = (\mathbf{k}_{\mathcal{X}}(\mathbf{x})^\top) \otimes (\mathbf{k}_{\mathcal{T}}(\mathbf{t})^\top)$.
For the conditional mean function we therefore have that
\begin{align*}
    m_n[f](\mathbf{x},\mathbf{t}) & = \mathbf{k}(\mathbf{x},\mathbf{t})^\top \mathbf{K}^{-1} f(X_n) \\
    & = [ (\mathbf{k}_{\mathcal{X}}(\mathbf{x})^\top) \otimes (\mathbf{k}_{\mathcal{T}}(\mathbf{t}))^\top ] [ (\mathbf{K}_{\mathcal{X}}^{-1}) \otimes (\mathbf{K}_{\mathcal{T}}^{-1}) ] f(X_n) \\
    & = [ \mathbf{k}_{\mathcal{X}}(\mathbf{x})^\top \mathbf{K}_{\mathcal{X}}^{-1} ] \otimes [ \mathbf{k}_{\mathcal{T}}(\mathbf{t})^\top \mathbf{K}_{\mathcal{T}}^{-1} ] f(X_n)
\end{align*}
where in the final equality we have exploited the \emph{mixed-product} property of the Kronecker product.
The dependence on $k_0^2$ of this expression occurs only in the term $\mathbf{k}_{\mathcal{X}}(\mathbf{x})^\top \mathbf{K}_{\mathcal{X}}^{-1}$.
From the calculations in \Cref{app: flat prior limit}, we can take the $k_0^2 \rightarrow \infty$ limit of the term $\mathbf{k}_{\mathcal{X}}(\mathbf{x})^\top \mathbf{K}_{\mathcal{X}}^{-1}$, to obtain that
\begin{align*}
m_n[f](\mathbf{x},\mathbf{t}) & = \left\{ \mathbf{k}_b(\mathbf{x})^\top \mathbf{K}_b^{-1} + [1 - \mathbf{k}_b(\mathbf{x})^\top \mathbf{K}_b^{-1} \mathbf{1}] \frac{\mathbf{1}^\top \mathbf{K}_b^{-1}}{\mathbf{1}^\top \mathbf{K}_b^{-1} \mathbf{1}} \right\} \otimes \left[ k_{\mathcal{T}}(\mathbf{t})^\top \mathbf{K}_{\mathcal{T}}^{-1} \right] f(X_n) .
\end{align*}
For the conditional covariance function, we have from a similar argument based on the mixed-product property that
\begin{align}
    k_n[f]((\mathbf{x},\mathbf{t}),(\mathbf{x}',\mathbf{t}')) & = k((\mathbf{x},\mathbf{t}),(\mathbf{x}',\mathbf{t}')) - \mathbf{k}(\mathbf{x},\mathbf{t})^\top \mathbf{K}^{-1} \mathbf{k}(\mathbf{x}',\mathbf{t}') \nonumber \\
    & = k_{\mathcal{X}}(\mathbf{x},\mathbf{x}') k_{\mathcal{T}}(\mathbf{t},\mathbf{t}') - \underbrace{[\mathbf{k}_{\mathcal{X}}(\mathbf{x})^\top \mathbf{K}_{\mathcal{X}}^{-1} \mathbf{k}_{\mathcal{X}}(\mathbf{x}')]}_{(*)}  [\mathbf{k}_{\mathcal{T}}(\mathbf{t})^\top \mathbf{K}_{\mathcal{T}}^{-1} \mathbf{k}_{\mathcal{T}}(\mathbf{t}')] . \label{eq: kn intermed}
\end{align}
The term $(*)$ can be read off from \eqref{eq: standard GP cov} and \eqref{eq: plug in k}:
\begin{align*}
    (*) = (k_0^2 \mathbf{1} + \mathbf{k}_b(\mathbf{x}))^\top \left\{ \mathbf{K}_b^{-1} - \frac{\mathbf{K}_b^{-1} \mathbf{1} \mathbf{1}^\top \mathbf{K}_b^{-1}}{k_0^{-2} + \mathbf{1}^\top \mathbf{K}_b^{-1} \mathbf{1}} \right\} (k_0^2 \mathbf{1} + \mathbf{k}_b(\mathbf{x}'))
\end{align*}
The fractional term can again be treated using the Taylor expansion \eqref{eq: Taylor for denominator} for $k_0^{-2}$ at 0, which upon expanding brackets yields 
\begin{align*}
    (*) = k_0^2 + \underbrace{ \mathbf{k}_b(\mathbf{x})^\top \mathbf{K}_b^{-1} \mathbf{k}_b(\mathbf{x}') - \frac{[\mathbf{k}_b(\mathbf{x})^\top \mathbf{K}_b^{-1} \mathbf{1} - 1] [\mathbf{k}_b(\mathbf{x}')^\top \mathbf{K}_b^{-1} \mathbf{1} - 1]^\top}{ \mathbf{1}^\top \mathbf{K}_b^{-1} \mathbf{1}}  + O(k_0^{-2}) }_{(**)} .
\end{align*}
Substituting this expression into \eqref{eq: kn intermed}, and using the definition of $k_{\mathcal{X}}$, we obtain that
\begin{align*}
    k_n[f]((\mathbf{x},\mathbf{t}),(\mathbf{x}',\mathbf{t}')) & = (k_0^2 + k_b(\mathbf{x},\mathbf{x}')) k_{\mathcal{T}}(\mathbf{t},\mathbf{t}') - (k_0^2 + (**)) [\mathbf{k}_{\mathcal{T}}(\mathbf{t})^\top \mathbf{K}_{\mathcal{T}}^{-1} \mathbf{k}_{\mathcal{T}}(\mathbf{t}')] \\
    & = k_0^2 \underbrace{ [k_{\mathcal{T}}(\mathbf{t},\mathbf{t}') - \mathbf{k}_{\mathcal{T}}(\mathbf{t})^\top \mathbf{K}_{\mathcal{T}}^{-1} \mathbf{k}_{\mathcal{T}}(\mathbf{t}')] }_{(***)} \\
    & \qquad + k_b(\mathbf{x},\mathbf{x}') k_{\mathcal{T}}(\mathbf{t},\mathbf{t}') - (**) [\mathbf{k}_{\mathcal{T}}(\mathbf{t})^\top \mathbf{K}_{\mathcal{T}}^{-1} \mathbf{k}_{\mathcal{T}}(\mathbf{t}')]
\end{align*}
where for $\mathbf{t}, \mathbf{t}'$ in the training set we have $(***) = 0$, which gives the desired result.

\section{Details for the Cardiac Model}
\label{app: cardiac}

This appendix summarises the most important aspects of the cardiac model that we used, and is principally intended for researchers working in the area of cardiac modelling who are interested in understanding the technical modelling aspects of the case study that we report in \Cref{sec: applications}.

\paragraph{Geometry}

The model was based on a heart geometry obtained from a computer tomography (CT) dataset of a single patient, used in \citet{strocchi2020simulating,strocchi2023cell}. 
The original patient geometry was recorded using an average mesh resolution of 1.06 $\pm$ 0.16 mm; below we describe how our coarser and finer meshes were constructed. 
Fibres in the atria and the ventricles were generated with \emph{universal atrial coordinates} \citep{roney2019universal} and a rule-based method developed by \citet{bayer2012novel}. 
Full details about the geometry are provided in \citet{strocchi2020simulating,strocchi2023cell}. 

\paragraph{Atrial and Ventricular Activation}

For the cardiac dynamics, we simulated atrial and ventricular activation using an \emph{Eikonal model} \citep{neic2017efficient}
\begin{align*}
\sqrt{\nabla t_a(\mathbf{x})^\top \mathbf{V}(\mathbf{x}) \nabla t_a(\mathbf{x})} &= 1 \qquad \mathbf{x}\in\Omega \; \\
t_a(\mathbf{x}) &= t_0 \qquad \mathbf{x}\in\Gamma
\end{align*}
 where $t_a(\mathbf{x})$ is the local activation time in the active domain $\Omega$, $\mathbf{V}$ is the tensor of the squared conduction velocities (CV) in the fibres, sheet and normal directions, and $\Gamma$ is a subset of nodes in the domain that get activated at time $t_0$.

\paragraph{Myocardium}

The cardiac model that we used treats the myocardium as a transversely isotropic conduction medium. 
The fibre CV in the atrial and ventricular myocardium was set to 0.6 m/s and 0.9 m/s, consistent with experimental measurements \citep{taggart2000inhomogeneous,hansson1998right}, while the anisotropy ratio was set to 0.4 in all tissues \citep{strocchi2023cell}. 
We simulated fast endocardial conduction in the ventricles by defining a 1 mm thick layer and assigning it with increased CV by five fold compared to normal ventricular myocardium \citep{ono2009morphological}. 
Fast conduction in the Bachmann bundle area in the atria was simulated by defining a region between the left and the right atrium as in \citet{strocchi2023cell}, and by assigning it with CV increased by 3.5 fold compared to normal atrial myocardium \citep{strocchi2023cell}. 
The atria and ventricles were electrically isolated by defining an insulating layer to avoid non-physiological activation and to control the atrio-ventricular delay. 
Atrial and ventricular activation were initiated at the site of the right atrial and right ventricular pacing lead, respectively, identified from the CT images. 
The atrio-ventricular delay was set to 150 ms. 

\paragraph{Active Tension}

Atrial and ventricular activation triggered a rise in local transmembrane potential \citep{neic2017efficient}, which then triggered a rise in active tension. Active tension was modelled with a phenomenological model developed by \citet{niederer2011length}:
\begin{align*}
S_a(\mathbf{x},t) &= T_{\text{peak}}\phi(\lambda)\tanh^2\Bigg(\frac{t_s}{\tau_r}\Bigg)\tanh^2\Bigg(\frac{t_{\text{dur}}-t_s}{\tau_d}\Bigg), \qquad 0<t_s<t_{\text{dur}}\;, \\
\phi(\lambda) &= \tan(\text{ld}(\lambda-\lambda_0))\;, \qquad t_{s} = t - t_{\text{a}}(\mathbf{x}) - t_{\text{emd}}\;
\end{align*}
where $T_{\text{peak}}$, $\tau_r$, $\tau_d$, $t_{\text{dur}}$, $\text{ld}$, $\lambda_0$ and $t_{\text{emd}}$ represent the reference tension, the rise and the decay time, the twitch duration, the reference stretch and the electro-mechanical delay. The following parameter values were used:
\begin{center} \small
\begin{tabularx}{0.7\textwidth}{|X|c|c|c|c|c|c|c|}
\hline
 & $T_{\text{peak}}$ & $\tau_r$ & $\tau_d$ & $t_{\text{dur}}$ & $\text{ld}$ & $\lambda_0$ & $t_{\text{emd}}$ \\
 & kPa & ms & ms & ms & - & - & ms \\
\hline
\textbf{Atria} & 80 & 100 & 50 & 450 & 0.7 & 6 & 20 \\
\textbf{Ventricles} & 60 & 50 & 50 & 200 & 0.7 & 6 & 20 \\
\hline
\end{tabularx}
\end{center}

\paragraph{Passive Properties}

We simulated passive properties of atrial and ventricular myocardium with a transversely isotropic Guccione law
\begin{align*}
\Psi(\mathbf{E}) &= \frac{C}{2}\big[e^{Q}-1\big]\;,\\
Q &= b_{\text{ff}}E_{\text{ff}}^2+2b_{\text{fs}}(E_{\text{fs}}^2 + E_{\text{fn}}^2) + b_{\text{ss}}(E_{\text{ss}}^2 + E_{\text{nn}}^2 + 2E_{\text{sn}}^2)\;,
\end{align*}
\noindent where f, s, n in the Cauchy--Green strain tensor $\mathbf{E}$ represent the strain in the local fibres, sheet and normal to sheet directions, and $C$, $b_{\text{ff}}$, $b_{\text{fs}}$ $b_{\text{ss}}$ are the bulk stiffness, and the stiffness in the fibre, cross-fibre and transverse directions. 
The following parameter values were used:
\begin{center} \small
\begin{tabularx}{0.5\textwidth}{|X|c|c|c|c|}
\hline
 & $C$ & $b_{\text{ff}}$ & $b_{\text{fs}}$ & $b_{\text{ss}}$ \\
 & kPa & - & - & - \\
\hline
\textbf{Atria} & 3 & 25 & 11 & 9 \\
\textbf{Ventricles} & 4 & 25 & 11 & 9 \\
\hline
\end{tabularx}
\end{center}
Parameters in tissues other than the atria and the ventricles were represented with a neo-Hookean law and the parameters were set according to \citet{strocchi2020simulating,strocchi2023cell}. 
For all tissues, near-incompressibility was enforced with a penalty method with a bulk modulus $\kappa=1000.0$ kPa. To constrain the motion of the heart, we applied omni-directional springs on the superior vena cava and the two right pulmonary veins, as well as a region around the apex (\Cref{fig:BCs}, left). Spring stiffness was set to 0.1 kPa/$\mu$m.

\paragraph{Adjusting the Spatial Resolution}

We used \texttt{meshtool}, open-source software for mesh manipulation \citep{neic2020automating}, to refine or coarsen the mesh to a desired spatial resolution for the experiments that we report. 
The spatial resolution $x_1$ that we used for extrapolation in the main text is the \emph{nominal} resolution fed into \texttt{meshtool}, but we note that the nominal resolution is typically not exactly achieved in the re-meshing process; we therefore also report the effective average mesh resolution: 
\begin{center} \small
\begin{tabular}{|c|c|c|c|}
\hline
\textbf{Target dx} [mm] & \textbf{Effective dx} [mm] & \textbf{\# nodes} & \textbf{\# elements} \\
\hline
1.7 & 1.7 & 119,104 & 540,621 \\
1.4 & 1.36 & 213,580 & 1,018,699 \\
1.0 & 1.06 & 417,863 & 1,988,945 \\
0.7 & 0.86 & 764,094 & 3,679,517 \\
0.6 & 0.7 & 1,498,007 & 7,512,728 \\
0.5 & 0.56 & 3,006,080 & 15,768,118 \\
0.4 & 0.43 & 6,217,838 & 30,699,422 \\
\hline
\end{tabular}
\end{center}

\paragraph{Computational Resources}

All simulations were carried out using the \emph{Cardiac Arrhythmia Research Package} (CARP) \citep{vigmond2003computational,augustin2021computationally} on ARCHER2, a UK national super computing service (\url{https://www.archer2.ac.uk/}).

\end{document}